\documentclass[journal]{IEEEtran}

\IEEEoverridecommandlockouts
\usepackage{epstopdf}
\usepackage{blindtext}
\usepackage{mathrsfs}
\usepackage{latexsym}
\usepackage{graphicx}
\usepackage{epsfig}
\usepackage{subfigure}
\usepackage{array}
\usepackage{amsmath}
\usepackage{amssymb}
\usepackage{color, soul}
\usepackage{amsthm}
\usepackage{enumerate}
\usepackage{algpseudocode}
\usepackage{algorithm}
\usepackage{bm}
\usepackage{amssymb}
\usepackage{balance}
\usepackage{verbatim}
\usepackage{setspace}

\theoremstyle{plain} 

\newtheorem{lemma}{Lemma}
\newtheorem{prop}{Proposition}
\theoremstyle{definition}

\newtheorem{remark}{Remark}

\def\bal#1\eal{\begin{align}#1\end{align}}

\newcommand{\bW}{{\bf W}}
\newcommand{\bH}{{\bf H}}

\newcommand{\bR}{{\bf R}}

\newcommand{\bY}{{\bf Y}}

\newcommand{\bI}{{\bf I}}
\newcommand{\bU}{{\bf U}}

\newcommand{\bA}{{\bf A}}

\newcommand{\bB}{{\bf B}}

\newcommand{\bQ}{{\bf Q}}

\newcommand{\bX}{{\bf X}}

\newcommand{\bs}{{\bf s}}

\newcommand{\gl}{\lambda}

\newcommand{\bx}{{\bf x}}

\newcommand{\ba}{{\bf a}}
\newcommand{\bb}{{\bf b}}
\newcommand{\bc}{{\bf c}}

\newcommand{\bh}{{\bf h}}
\newcommand{\bz}{{\bf z}}
\newcommand{\bo}{{\bf 0}}

\newcommand{\bw}{{\bf w}}
\newcommand{\bRequire}{{\bf Require}}

\newcommand{\bp} {\begin{proof}}
\newcommand{\ep} {\end{proof}}

\newcommand{{\Rb}} {\right)}

\newcommand{{\Rf}} {\right\}}
\newcommand{{\diag}} {\mathrm{diag}}

\begin{document}

\title{Secure Cognitive Radio Communication  via Intelligent Reflecting Surface} 

\author{Limeng Dong, Hui-Ming Wang \emph{Senior Member, IEEE},  and Haitao Xiao


\thanks{The authors are with the School of Information and Communications Engineering, and also with the Ministry of Education Key Lab for Intelligent Networks and Network Security, Xi'an Jiaotong University, Xi'an, Shaanxi 710049, China (e-mail: dlm$\_$nwpu@hotmail.com; xjbswhm@gmail.com; xht8015949@xjtu.edu.cn)}

}

\maketitle

\begin{abstract}


In this paper,  an intelligent reflecting surface (IRS) assisted  spectrum sharing underlay cognitive radio (CR) wiretap channel (WTC) is studied, and we aim at enhancing the secrecy rate of secondary user in this channel subject to total power constraint at secondary transmitter (ST), interference power constraint (IPC) at primary receiver (PR) as well as unit modulus constraint at IRS. Due to extra IPC and eavesdropper (Eve) are considered, all the existing solutions for enhancing secrecy rate of  IRS-assisted non-CR WTC as well as enhancing transmission rate in IRS-assisted CR channel without eavesdropper fail in this work. Therefore, we propose new numerical solutions to optimize the secrecy   rate of this channel  under full primary, secondary users' channel state information (CSI) and three different cases of Eve's CSI: full CSI, imperfect  CSI with bounded estimation error, and no  CSI.  To solve the difficult non-convex  optimization problem, an efficient alternating optimization (AO) algorithm is proposed to jointly optimize the beamformer at ST and phase shift coefficients at IRS. In particular, when optimizing the phase shift coefficients during each iteration of AO, a Dinkelbach based solution in combination with successive approximation and penalty based solution is proposed under full CSI and a penalty convex-concave procedure solution is proposed under imperfect Eve's CSI. For no Eve's CSI case, artificial noise (AN) aided approach is adopted to help enhancing the secrecy rate. Simulation results show that our proposed solutions for the IRS-assisted design greatly enhance the secrecy performance compared with the existing numerical solutions with and without IRS under full and imperfect Eve's CSI. And positive secrecy rate can be achieved by our proposed AN aided approach given most channel realizations  under no Eve's CSI case so that secure communication also  can be guaranteed. All of the proposed  AO algorithms are guaranteed to monotonic convergence.

\end{abstract}	


\section{Introduction}

Due to exponential growth of wireles systems and services in the past two decades, spectrum has become a very scarce resource. Cognitive radio (CR) concept  has been proposed and is treated as one of the most promising technologies. Using the spectrum sensing and sharing technology, the  growth of wireless devices and the scarcity of spectrum resources can be effectively alleviated.  However,  CR networks is suffering from lots of security threads, such as primary user emulation, jamming and eavesdropping in the physical layer, and spectrum sensing data falsification in the network layer as well as cross layer attacks \cite{Fragkiadakis-13}. Eavesdropping  is a kind of  passive attack in the physical layer that  brings great security risks  due to the broadcast nature of wireless channels. Physical layer security (PLS) approach has emerged as a very valuable complement to cryptography-based approaches. The key idea of this approach is that by exploiting  the properties of wireless channels,  the transmitted information for legitimate users can be completely “hidden” from eavesdropping, resulting secrecy of communication \cite{Bloch-11}. And the secrecy rate, namely as the difference of mutual information of legitimate user and eavesdropper, is the key performance metric for secrecy communications, and various works were established to enhance the secrecy rate of CR channels analytically or numerically.

\subsection{Related work}

A number of research results for secure  multi-antenna spectrum sharing CR communication based on PLS were established \cite{Pei-10}-\cite{Fang-15}. In \cite{Pei-10}, the secrecy rate maximization  of CR multi-input single-output (MISO) wiretap channel (WTC) under full  channel state information (CSI) subject to total power constraint (TPC) at secondary transmitter as well as interference power constraint (IPC) at primary receiver were  studied, and the imperfect CSI case was also established in \cite{Pei-11}. In \cite{Nahari-11}-\cite{Wang-15},  artificial noise (AN) aided approach was also proposed to maximize the secrecy rate. Apart from the MISO case, secure CR multi-input multi-output (MIMO) WTC was also considered in \cite{Dong-18}-\cite{Fang-15}, and  the secrecy rate maximization of this channel was investigated either analytically or numerically.    However, there is an inevitable problem in  spectrum sharing CR networks: as primary receiver is located close to the secondary transmitter, the IPC is becoming  tight so that the transmitter is unable to allocate full power for signaling. Therefore, the achievable secrecy rate at secondary user is likely to saturate with transmit power \cite{Dong-18}. Hence, how to eliminate the restrictions brought by IPC on enhancing the secrecy performance in CR channels is still an open problem.
  
 Recently, intelligent reflecting surface (IRS), has been proposed and it has drawn wide attention due to its extraodinary advantages. IRS is a   metasurface consisting of low complexity passive reflecting elements \cite{Hu-18}. These elements could change the propagation channels by  inducing certain phase shift  for the incident electromagnetic signal waves via the software in the controller so that  the quality of communications at user can be greatly improved. Since IRS  is a passive device,  it does not consume any power for signal reflection, and can be easily deployed on many areas such as buildings, ceilings or indoor spaces. Furthermore,  IRS  also does not produce any extra noise to users since it is not equipped with A/D, D/A converter, power amplifiers  or other signal processing devices. Therefore, these great benefits make IRS as a promising green energy-efficient technique in beyond 5G or even 6G communications \cite{Yuan-20}-\cite{Renzo-20b}.

Inspired by these advantages brought by IRS, several  studies of IRS-assisted multi-antenna communication  were  proposed, and lots of results  were shown that IRS greatly enhance  the transmission rate under either full CSI \cite{Huang-20}-\cite{Pan-20} or imperfect CSI \cite{Hu-20}\cite{Zhou-19}. In \cite{Huang-19c}, it was shown that the IRS-assisted design also greatly improves the energy efficiency. Motivated by these works, IRS was also applied to enhance the transmission in multi-antenna CR channels \cite{Yuan-19}-\cite{Xu-20b}. In \cite{Yuan-19}-\cite{Xu-20}, the transmission rate optimization of IRS-assisted MISO CR channel subject to TPC, IPC and unit modulus constraint (UMC) was studied, and it was later extended to MIMO case \cite{Zhang-20}. In \cite{He-20}\cite{Zhang-20b}, the power minimization optimization algorithm of IRS-assisted MISO CR channel subject to IPC, UMC as well as target quality-of-service (QoS) for secondary user under full and statistical CSI were proposed. The sum rate of a full-duplex IRS-assisted MISO CR channel was also investigated in \cite{Xu-20b}. All these works again validated that IRS greatly improved the system performance of CR channels.

 Furthermore, IRS was also combined with PLS to enhance the secrecy of communication in MISO/MIMO WTCs \cite{Cui-19}-\cite{Dong-20c}. 
 In  \cite{Cui-19}-\cite{Yu-19}, secure IRS-assisted MISO WTC subject to TPC at transmitter and UMC at IRS was studied, and the simulation results showed that proposed numerical solutions for this IRS-asssited desgin greatly boosted the  user's secrecy rate compared with existing solutions for no IRS case. In \cite{Dong-21},  a double IRS-assisted design was considered, and  a product Riemmanian manifold based algorithm was investigated  to solve the non-convex secrecy rate problem. In \cite{Dong-20d}, an efficient AN-aided method was proposed to enhance the secrecy rate when there is completely no eavesdropper's CSI. In \cite{Xu-19}, the IRS-assisted multi-user MISO downlink wiretap channel was also investigated under full eavesdropper's CSI, and it was later extended to the case of imperfect eavesdropper's CSI  \cite{Yu-20}. In addition to MISO case, the numerical solutions for enhancing the
 IRS-assisted MIMO WTC was recently established  \cite{Dong-20}\cite{Dong-20c}. Same with the non-secure case, all these current literatures for secure IRS-assisted multi-antenna system indicated that IRS greatly boosted the secrecy rate compared with traditional solutions.

However, all the studies \cite{Cui-19}-\cite{Dong-20c} for secure IRS-assisted design didn't consider CR-setting in which extra IPC is involved in addition to TPC and UMC. Therefore, all the solutions  in the current literatures may fail to CR case (unless the IPC is a relaxed ignorable constraint). 
  Furthermore, all the existing solutions for the IRS-assisted CR channels \cite{Yuan-19}-\cite{Xu-20b} did  not consider security issues. When eavesdropper exists in the system, the secrecy rate optimization problem becomes complicated due to the new structure of objective function, which is significantly more complex than the single log formular for the non-secure case.  Hence,  new efficient numerical solutions is necessary  to develop in the  secure CR setting. 

\subsection{Contributions}

 Against the above background, in this paper, we consider an IRS-assisted spectrum sharing  underlay CR MISO WTC, and  focus on  enhancing  the secrecy rate at secondary user subject to TPC at secondary transmitter, IPC at primary receiver as well as UMC at IRS. The key motivation for applying IRS to secure CR system in this work is of two aspects: firstly, to the best of our knowledge, the study of IRS-assisted secure CR communication  has never appeared in the literature, and all the existing numerical solutions cannot be directly applied to this setting; secondly, with IRS, full power allocation at secondary transmitter for signaling can be realized so that the secrecy rate can grow unbounded  with transmit power. This is significantly different from those studies without IRS \cite{Pei-10}-\cite{Fang-15} in which the  IPC becomes tight  as the transmit power increases so that the secrecy rate saturates eventually. Specifically,  we assume that full CSI of legitimate primary and secondary user are available, and consider three conditions of   eavesdropper's CSI: full CSI, imperfect CSI with bounded estimation error and completely no CSI. And  numerical algorithms are developed to enhance the secrecy rate  under each considered CSI condition.  The following summarizes the key contributions of our work.

1). Firstly, we assume that the eavesdropper's CSI  is perfectly  available, and propose  an efficient  alternating optimization (AO) algorithm to jointly optimize the beamformer $\bw$ and phase shift coefficients $\bs$ in the non-convex secrecy rate optimization problem. The main dificulty of this work is of optimizing the $\bs$ given $\bw$ in the non-convex fractional programming sub-problem. In particular, we propose a Dinkelbach method in combination with successive convex approximation (SCA) and  penalty based (PB) approach to optimize the secrecy rate. As the convergence of AO algorithm is reached, a limit point solution for the problem can be obtained.

2). Secondly, we assume that the eavesdropper's CSI  is imperfectly available due to the bounded channel estimation errors. To solve the complicated non-convex secrecy rate optimization problem,  we firstly transform the infinite non-convex constraints to convex one by adding auxiliary variables, and  then apply AO algorithm to jointly optimize  $\bw$ and $\bs$. When optimizing $\bw$ given $\bs$, SCA algorithm is applied to solve the sub-problem, and when optimizing $\bs$ given $\bw$,  a penalty convex concave procedure (P-CCP) approach is proposed. The key idea of P-CCP is to relax the problem by adding slack variables so that the non-convex UMC can be violated, and then penalizing the sum of the violations.  As the convergence is reached, the solution returned by the P-CCP is an approximate first-order optimal solution for the original sub-problem.

3). Thirdly, we assume a more practical case that the eavesdropper's CSI is completely unavailable at transmitter. To enhance the secrecy performance under this case, we propose an artificial noise (AN) aided approach, in which a minimum power is firstly optimized to meet a target QoS  constraint at secondary receiver, and then  the residual power at secondary transmitter is applied for AN signaling so as to decrease the quality of communication at eavesdropper. Furthermore, to guarantee that the target QoS at secondary receiver is  not affected by the AN, we set the directions for AN signaling to both null space of the aggregated   channels of secondary and primary  receiver so that only eavesdropper is interfered by AN. 

4). Simulation results have validated the monotonic convergence of the proposed algorithm. And when full CSI and imperfect Eve's CSI are assumed, our proposed AO algorithm for the IRS-assisted design greatly enhances the secrecy rate compared with the existing solutions for no IRS case. And even when the primary receiver is located close to secondary transmitter,  allocating full power for signaling is still possible and the secrecy rate can grow unbounded with transmit power as if there is no IPC. This is significantly different from no IRS case in which the secrecy rate is likely to saturate since full power cannot be allocated for signaling. When no Eve's CSI is assumed,  positive secrecy rate  can be achieved by our proposed AN aided scheme under most scenarios and hence secure communications also can be guaranteed.

The rest of the paper is organized as follows: Section II describes the channel model. In Section III, the algorithm for secrecy rate maximization under full Eve's CSI is proposed. In Section IV, the algorithm for secrecy rate maximization under imperfect eavesdropper's CSI is proposed. Section V gives solution for enhancing the secrecy rate under no eavesdropper's CSI. Simulation results have been carried out to evaluate the performance and convergence of proposed algorithm in Section VI. Finally, Section VII concludes the paper.

\emph{Notations}: Bold lower-case letters ($\ba$) and capitals ($\bA$) denote the vector and matrix respectively; $\bA^{\rm T}$, $\bA^{*}$ and $\bA^{\rm H}$ denote transpose, conjugate and  Hermitian conjugate of $\bA$, respectively; $\bA \geq \bo$ means  positive semi-definite; $E\left \{ \cdot  \right \}$ is statistical expectation, $\gl_i(\bA)$ denotes eigenvalues of $\bA$, which are in decreasing order unless indicated otherwise, i.e. $\gl_1\ge \gl_2\ge \gl_3....$; $tr(\bA)$ and $|\bA|$ are the trace and determinant of $\bA$; $\bI$ is an identity matrix with appropriate size; $\mathbb{C}^{M \times N}$ and $\mathbb{R}^{M \times N}$ are the space of $M\times N$ matrix with complex-valued elements and real-valued elements, respectively;  $diag(\ba)$ is to transform the vector $\ba$ to a diagonal matrix in which all diagonal entries are in $\ba$; $arg\{\ba\}$ denotes the phase of each entry of   $\ba$; $rank(\bA)$ is the rank of $\bA$; $\parallel \ba \parallel$ denotes the Euclidian norm of vector $\ba$, and $\parallel \bA \parallel_F$ denotes the F-norm of matrix $\bA$;  $Re\{a\}$ denotes the real element of a; $\odot$ denotes Hadamard product; $\mathcal{N}(A)$ denotes the null space of $\bA$.

\section{Channel Model}

Considering an IRS-assisted spectrum sharing underlay  CR Gaussian WTC model shown in Fig.1, in which an IRS, a  secondary transmitter (Alice),  secondary receiver (Bob),  eavesdropper (Eve) as well as a primary receiver (PR) are included.  In this model, Alice is equipped with $m$ antennas and Bob, Eve, PR are all equipped with single antenna, the IRS is equipped with $n$ reflecting elements. To help Alice enhancing the secrecy performance, the task of IRS in this model is to adjust the phase of the incident signals via the reflecting elements. And  there is a controller   connected with both IRS and Alice which is used to control the phase shift coefficient of reflecting elements as well as other coordination tasks  for channel acquisition and data transmission  \cite{Cui-19}. 

\begin{figure}[t]
	\centerline{\includegraphics[width=3.2in]{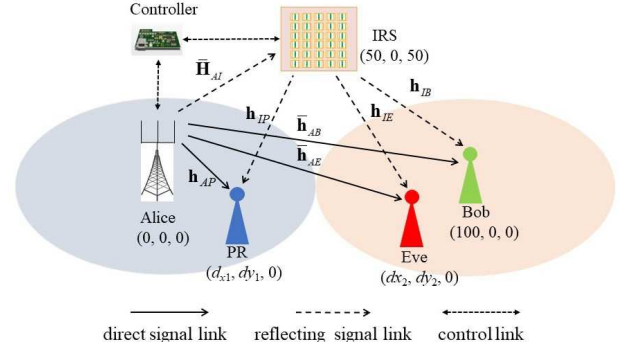}}
	\caption{A block diagram of IRS-assisted Gaussian CR MISO wiretap channel}
\end{figure}

Denote $x$ as the information signals to Bob following $E\left \{ |x|^2  \right \}=1$, let $\bar{\bh}_{AB}\in\mathbb{C}^{1\times m}$, $\bar{\bh}_{AE}\in\mathbb{C}^{1\times m}$, $\bh_{IB}\in\mathbb{C}^{1\times n}$,  $\bh_{IE}\in\mathbb{C}^{1\times n}$, $\bar{\bH}_{AI}\in\mathbb{C}^{n\times m}$, $\bh_{AP}\in\mathbb{C}^{1\times m}$ and $\bh_{IP}\in\mathbb{C}^{1\times n}$ be as the communication link of Alice-Bob, Alice-Eve,   IRS-Bob, IRS-Eve, Alice-IRS, Alice-PR and IRS-PR respectively, then the received signals at Bob $y_B$ and Eve $y_E$ are
\bal
\notag
&y_B= (\bar{\bh}_{AB}+\bh_{IB}diag(\bs^*)\bar{\bH}_{AI})\bw x+\xi_B,\\
\notag
&y_E= (\bar{\bh}_{AE}+\bh_{IE}diag(\bs^*)\bar{\bH}_{AI})\bw x+\xi_E
\eal
respectively where $\bw\in\mathbb{C}^{m\times 1}$ is the beamformer at Alice, $\bs=[e^{j\theta_1},e^{j\theta_2},...,e^{j\theta_n}]^{\rm H}$, $\theta_i$ is the phase shift coefficient at the $i$-th  reflecting element, $\xi_{B}\sim\mathcal{CN}(0,\sigma ^2_B)$ and $\xi_{E}\sim\mathcal{CN}(0,\sigma ^2_E)$  represent complex noise at  Bob and Eve respectively, $\sigma ^2_B$, $\sigma ^2_E$ denote the noise power. After some manipulations, $y_B$ and $y_E$ can be further equivalently expressed as
\bal
\notag
&y_B=(\bar{\bh}_{AB}+\bs^{\rm H}\bar{\bH}_B)\bw x+\xi_B,\\
\notag
&y_E=(\bar{\bh}_{AE}+\bs^{\rm H}\bar{\bH}_E)\bw x+\xi_E
\eal
respectively, where $\bar{\bH}_B=diag(\bh_{IB}^T)\bar{\bH}_{AI}$, $\bar{\bH}_E=diag(\bh_{IE}^T)\bar{\bH}_{AI}$ denote the effective cascaded channels of the reflecting communication link Alice-IRS-Bob, Alice-IRS-Eve respectively. We assume that the CSI of $\bar{\bh}_{AB}$, $\bar{\bH}_B$ are fully known at Alice throughout the paper.  Based on this signal model, the transmission rate $C_B$ at Bob and $C_E$ at Eve are expressed as 
\bal
\notag
&C_B=\log_2(1+|(\bh_{AB}+\bs^{\rm H}\bH_B)\bw|^2),\\
\notag
&C_E=\log_2(1+|(\bh_{AE}+\bs^{\rm H}\bH_E)\bw|^2),
\eal
 respectively, where  for $i\in\{B, E\}$, $\bh_{Ai}=\bar{\bh}_{Ai}/\sigma_i$, $\bH_i=\bar{\bH}_i/\sigma_i$. Since primary user is the licensed spectrum holder in spectrum sharing network, the interference generated from Alice to PR should keep below a certain pre-defined threshold. Therefore, the design of $\bw$ and $\bs$ should satisfy both TPC and IPC, i.e., $\parallel \bw \parallel^2\leq P_T$ and $|(\bh_{AP}+\bs^{\rm H}\bH_P)\bw|^2\leq P_I$, where $P_T$ is the total power budget at Alice, $P_I$ is the maximum interference power threshold at PR and where $\bH_P=diag(\bh_{IP}^{\rm T})\bH_{AI}$ is the effective cascaded channel of the reflecting communication link Alice-IRS-PR. Since PR is also a legitimate user, we assume that PR shares its CSI of $\bH_P$ and $\bh_{AP}$ with Alice. IRS is only capable of adjusting the phase of the signals without changing the amplitudes so that $\bs$ follows the UMC $|s_i|=1$ where $s_i$ is the $i$-th entry of $\bs$. Furthermore,  the  interference power generated by primary transmitter to Bob and Eve are treated as complex Gaussian noise mixed  in  $\xi_B$ and $\xi_E$ , and secondary user has no privilege to make primary transmitter changes its signaling strategy, which are also  standard settings being widely accepted in the literature  \cite{Pei-10}-\cite{Fang-15}, \cite{Yuan-19}-\cite{Xu-20b}.

Therefore, based on the key concept of information-theoretic PLS, to guarantee secure communication for this channel, the achievable secrecy rate $C_B-C_E$ should be positive, and larger secrecy rate indicates better secrecy performance \cite{Bloch-11}. In the following sections, we  focus on  enhancing the secrecy performance of this IRS-assisted channel by jointly optimizing $\bw$ and $\bs$ under three conditions of the Eve's CSI $\bH_E$ and $\bh_{AE}$: full CSI,  imperfect CSI with bounded estimation error, and completely no CSI.

\section{Algorithm for maximizing secrecy rate under full Eve's CSI}

Firstly, we consider an ideal assumption that Eve's CSI $\bar{\bh}_{AE}$ in the direct link and $\bar{\bH}_E$ in the reflecting link are perfectly known at Alice. This can be realized since Eve is just other user in the system and it also shares its CSI with Alice but is untrusted by Bob. For how to estimate the CSI of all the channel links, we consider that the channel is with quasi-static block fading, and focus on one particular fading block with length $T$ symbols over which all the channels remain approximately constant. Then,  $T$ is divided into three successive time slots, and Bob, Eve, PR send pilot signals  to Alice in each  time slot  so as to estimate the CSI. Finally, each time slot can be further divided into two smaller time slots, and we apply  the existing channel estimation solutions (see e.g. \cite{Wang-19}) to estimate the direct and cascaded channel in each smaller time slot.

Given full CSI, we formulate the secrecy rate optimization problem for this IRS-assisted CR MISO WTC as follows.
\bal
\notag
&(P1)\ \underset{\bw,  \bs}{\max} \  C_B-C_E\\
\label{TPC+IPC}
s.t. &\ \parallel \bw \parallel^2\leq P_T, |(\bh_{AP}+\bs^{\rm H}\bH_P)\bw|^2\leq P_I, \\ 
\label{UMC}
\quad\quad &|s_i|=1, i=1,2,...,n.
\eal
$P1$ is a non-convex optimization problem due to non-convex objective function and constraints, and currently there is no closed-form or numerical solutions for this problem in the literature. Since the variable $\bw$ and $\bs$ are coupled in both the objective function and IPC constraints, it is difficult to simultaneously optimize them. Therefore, in this subsection, we propose an AO algorithm  to jointly optimize $\bw$ and $\bs$ in two sub-problems.

Firstly, given fixed $\bs$, we formulate the sub-problem for optimizing $\bw$  as
\bal
\notag
(P2)\ \underset{\bw}{\max} \  C_B-C_E \ \ \ s.t.\ \  \eqref{TPC+IPC}.
\eal 
It can be known that $P2$ is a standard secrecy capacity of CR Gaussian wiretap channel optimization problem, which can be directly solved via the existing solutions in \cite{Pei-10}. By setting $\bR=\bw\bw^{\rm H}$, and $\bh_B=\bh_{AB}+\bs^{\rm H}\bH_B$, $\bh_E=\bh_{AE}+\bs^{\rm H}\bH_E$, $\bh_P=\bh_{AP}+\bs^{\rm H}\bH_P$,  P2 can be expressed as
\bal
\notag
(P2')\  &\underset{\bR}{\max} \   (1+\bh_B\bR\bh_B^{\rm H})(1+\bh_E\bR\bh_E^{\rm H})^{-1}\\
\notag
 &s.t.\quad tr(\bR)\leq P_T, \bh_P\bR\bh_P^{\rm H} \leq P_I, \bR\geq\bo.
\eal
 To solve $P2'$, the existing Charnes-Cooper transformation illustrated in \cite{Pei-10} can be directly applied to globally optimize $\bR$, from which the details is omitted here.

After optimizing $\bw$ given $\bs$, the next step is to optimize $\bs$ in the following sub-problem
\bal
\notag
(P3)\ \underset{\bs}{\max}\ h_B(\bs)/h_E(\bs), \ s.t. \ \eqref{UMC}, h_P(\bs)\leq P_I,
\eal
where for $j\in\{B, E, P\}$, $h_j(\bs)=1+\bs^{\rm H}\bar{\bH}_{Ij}\bs+2Re\{\bs^{\rm H}\tilde{\bh}_{Ij}\}+\bar{h}_{Aj}$, 
$\bar{h}_{Aj}=\bh_{Aj}\bR\bh_{Aj}^{\rm H}, \tilde{\bh}_{Ij}=\bH_j\bR\bh_{Aj}^{\rm H}, 
\bar{\bH}_{Ij}=\bH_j\bR\bH_j^{\rm H}$.
$P3$ is a fractional programming  problem, and we propose Dinkelbach based method \cite{Shen-19}\cite{Dinkelbach-67} in combination with SCA and PB algorithm to solve this problem. The key idea of this algorithm is summarized as three aspects: firstly, using Dinkelbach method to transform the problem to a more tractable non-fractional programming problem by adding extra variable; secondly, using SCA to approximate the non-convex objective function and constraints in the transformed problem to simpler linear formular; finally, PB  approach is adopted to globally solve this approximated problem.

\begin{remark}
Here we point out that the Dinkelbach method should be applied ahead of the SCA method when optimizing $P3$. The main reason is that the there are extra terms 1, $2Re\{\bs^{\rm H}\tilde{\bh}_{IB}\}$,  $\bar{h}_{AB}$ in the numerator of $h(\bs)$, and extra terms 1, $2Re\{\bs^{\rm H}\tilde{\bh}_{IE}\}$,  $\bar{h}_{AE}$ in the denominator of $h(\bs)$, which  fail to directly  approximate $h(\bs)$ using  Lemma 1 in \cite{Guan-20}. Hence, motivated by this issue, the solution  we  apply is to firstly transform the problem via Dinkelbach method to a new one, and then use SCA to approximate the new problem to more tractable form (see e.g. \cite{Shen-19}). Although a first-order optimal solution of P3 is out of reach by this method, the algorithm is guaranteed to converge so that a limit point solution can be obtained.

\end{remark}

Specifically, by adding a non negative $u$, $P3$ can be transformed to
\bal
\notag
(P4)\ &\underset{\bs}{\min}\ f(\bs/u)=h_E(\bs)-uh_B(\bs),\ s.t.\ \ \eqref{UMC}, h_P(\bs)\leq P_I.
\eal
Let $\bs(u)$ be as the optimal value of $P4$ given fixed $u$,  according to the key idea of Dinkelbach method  \cite{Dinkelbach-67}, finding the optimal solution of $P3$ is equivalent to searching for the optimal $u$ such that the value of objective function $f(\bs(u)/u)=0$, which can be obtained via BS algorithm since $f(\bs(u)/u)$ is  monotonically decreasing in $u$. 

$P4$ is still  non-convex due to the objective function and non-convex IPC. In the following, we propose SCA method to deal with $P4$, in which the key idea is  to firstly approximate the non-convex (concave) objective function and constraint to more simpler linear formular given a feasible starting point, and then iteratively optimize the  approximated problem. As the convergence is reached,  the  solutions returned by SCA  is a  first-order optimal point for the original problem \cite{Huang-19c}\cite{Sun-17}. 

To approximate the objective function and IPC in $P4$,  the following key lemma is needed, from which the proof can be seen in \cite{Song-15}.
\begin{lemma}
Let $\bA$ be an $n \times n$ Hermitian matrix, then for any point $\tilde{\bx}\in\mathbb{C}^{n\times 1}$, $\bx^{\rm H}\bA\bx$ is upper bounded by $\bx^{\rm H}\bA\bx\leq\bx^{\rm H}\lambda_1(\bA)\bI\bx-2Re\{\bx^{\rm H}(\lambda_1(\bA)\bI-\bA)\tilde{\bx}\}+\tilde{\bx}^{\rm H}(\lambda_1(\bA)\bI-\bA)\tilde{\bx}$.
\end{lemma}

Therefore, denote $\tilde{\bs}$ as a feasible point for $P4$,  using this lemma, the objective function $f(\bs/u)$ can be linearly approximated as
\bal
\notag
f(\bs/u)&\leq c_1(u)+c_2(u)\\
\notag
&+2Re\{\bs^{\rm H}[\tilde{\bh}_{EB}(u)-(\lambda_1(\bar{\bH}_{EB}(u))\bI-\bar{\bH}_{EB}(u))\tilde{\bs}]\}
\eal
where $\tilde{\bh}_{EB}(u)=\tilde{\bh}_{IE}-u\tilde{\bh}_{IB}$, $\bar{\bH}_{EB}(u)=\bar{\bH}_{IE}-u\bar{\bH}_{IB}$, 
$c_1(u)=1+\bar{h}_{AE}-u(1+\bar{h}_{AB})+n\lambda_1(\bar{\bH}_{EB}(u)), 
c_2(u)=n\lambda_1(\bar{\bH}_{EB}(u))-\tilde{\bs}^{\rm H}\bar{\bH}_{EB}(u)\tilde{\bs}$. 
 Similarly, the IPC in $P4$ can also be linearly approximated as
$h_P(\bs)\leq c_3+\tilde{h}_P(\bs/\tilde{\bs})\leq P_I$ 
where 
$c_3=\bar{h}_{AP}+2n\lambda_1(\bar{\bH}_{IP})-\tilde{\bs}^{\rm H}\bar{\bH}_{IP}\tilde{\bs}$, 
  $\tilde{h}_P(\bs/\tilde{\bs})=2Re\{\bs^{\rm H}[\tilde{\bh}_{IP}-(\lambda_1(\bar{\bH}_{IP})\bI-\bar{\bH}_{IP})\tilde{\bs}]\}$.
Hence, after dropping the constant terms, $P4$ can be approximated to a new optimization problem $P5$.
\bal
\notag
(P5)\ &\underset{\bs}{\min}\ \tilde{f}(\bs/(u,\tilde{\bs}))=Re\{\bs^{\rm H}[\tilde{\bh}_{EB}(u)-(\lambda_1(\bar{\bH}_{EB}(u))\bI\\
\notag
&-\bar{\bH}_{EB}(u))\tilde{\bs}]\}\ s.t.\ \eqref{UMC},\ \tilde{h}_P(\bs/\tilde{\bs})\leq P_I-c_3=\tilde{P}_I.
\eal
Therefore, by setting a feasible starting point $\tilde{\bs}$ and iteratively solve $P5$, the solution of $\bs$ is treated as a new starting point $\tilde{\bs}$. As the convergence is reached, the solution returned by SCA algorithm is a first-order optimal point for $P4$. 

However, $P5$ is still difficult to directly solve due to  non-convex \eqref{UMC}. But note that even when the problem is non-convex,  the complementary slackness condition 
\bal
\label{optimal condition}
\mu(\tilde{h}_P(\bs(\mu)/\tilde{\bs})-\tilde{P}_I)=0.
\eal
is always the necessary condition for optimality \cite{Boyd-04}, where  $\mu$ is the Lagrange multiplier respect to IPC $\tilde{h}_P(\bs(\mu)/\tilde{\bs})\leq\tilde{P}_I$. Hence, based on this key property, the optimal solution $\bs$ for $P5$ can be optimized based on two cases.

Considering the first case that $\mu=0$, i.e.,  IPC is an inactive constraint so that $P5$ reduces to the following $P5'$
\bal
\notag
(P5'): \ &\underset{\bs}{\max}\  \  -\tilde{f}(\bs/(u,\tilde{\bs}))\ s.t. \ \eqref{UMC}.
\eal
Hence, it can be directly known that the global optimal solution of $P5'$ is 
\bal
\label{optimal_inactive}
\bs(\mu=0)=arg\{[\lambda_1(\bar{\bH}_{EB}(u))\bI-\bar{\bH}_{EB}(u)]\tilde{\bs}-\tilde{\bh}_{EB}(u)\}.
\eal

Next, we consider the second case that $\mu>0$, i.e., IPC is a tight active constraint. Note that when IPC is tight, it is difficult to directly obtain a closed-form solution of $\bs$ as  in the first case. Therefore, inspired by the solution illustrated in \cite{Pan-20}, we propose a PB approach to globally optimize $\bs$ in this case. The key idea of this approach is of two aspects: firstly, we keep \eqref{UMC} unchanged and  transform $P5$ to a new problem in which the IPC is absorbed by the objective function via adding a penalty variable; secondly, we give the closed-form solution of $\bs$ in the transformed problem given fixed penalty variable. If a proper value of penalty variable is found, the corresponding solution of $\bs$ for the transformed problem is a global optimal solution for the original problem $P5$. Specifically,  we set the dual variable $\mu>0$ for the IPC as a penalty variable, and transform $P5$ to a new problem $P5''$ as follows.
\bal
\notag
(P5'') \ &\underset{\bs}{\max}\  -\tilde{f}(\bs/(u,\tilde{\bs}))-\mu\tilde{h}_P(\bs(\mu)/\tilde{\bs})\  \ s.t. \ \eqref{UMC},
\eal
Then, the global optimal solution of $P5''$ can be obtained as
\bal
\notag
&\bs(\mu>0)=arg\{[(\lambda_1(\bar{\bH}_{EB}(u))\bI-\bar{\bH}_{EB}(u)]\tilde{\bs}\\
\label{optimal_activeIPC}
&-\tilde{\bh}_{EB}(u)-\mu(\tilde{\bh}_{IP}-(\lambda_1(\bar{\bH}_{IP})\bI-\bar{\bH}_{IP})\tilde{\bs})\}.
\eal
To find the optimal value of $\mu$, the following lemma is needed
\begin{lemma}
 Let  $\bs(\mu)$ be as the optimal solution of $P5$ given $\mu$, then  $\tilde{h}_P(\bs(\mu)/\tilde{\bs})$ is  monotonically non-increasing in $\mu$.
\end{lemma}
\begin{proof}
Specifically, the Lagrangian of $P5$ is as
$L(\bs, \mu, v_i)=\tilde{f}(\bs/(u,\tilde{\bs}))+\mu(\tilde{h}_P(\bs/\tilde{\bs})-\tilde{P}_I)+\sum_{i=1}^{n}v_i(|q_i|-1)$,
where $\mu$ and $v_i$ denote the Lagrange multipliers responsible for the IPC and UMC respectively. Considering that $\mu_1>\mu_2>0$, and let $\bs(\mu_1)$,  $\bs(\mu_2)$ denote the  solution of $P5$ with $\mu_1$, $\mu_2$. Assume $\bs(\mu_1)$ and  $\bs(\mu_2)$ are both the global optimal solution of $P5$, then
$L(\bs(\mu_1), \mu_1, v_i)\leq L(\bs(\mu_2), \mu_1, v_i), L(\bs(\mu_2), \mu_2, v_i)\leq L(\bs(\mu_1), \mu_2, v_i)$.
Combining these two inequalities, one obtains
$(\mu_1-\mu_2)\tilde{h}_P(\bs(\mu_1)/\tilde{\bs})\leq(\mu_1-\mu_2)\tilde{h}_P(\bs(\mu_2)/\tilde{\bs})$.
Since $\mu_1>\mu_2$, then $\tilde{h}_P(\bs(\mu_1)/\tilde{\bs})\leq \tilde{h}_P(\bs(\mu_2)/\tilde{\bs})$, from which the lemma holds.
\end{proof}

Therefore, BS algorithm can be applied to optimize the optimal $\mu$ satisfying $\tilde{h}_P(\bs(\mu)/\tilde{\bs})=\tilde{P}_I$. The global convergence of this proposed method is analyzed as follows.
\begin{prop}
 If IPC is tight and active, the output solution $\bs(\mu)$ for $P5''$ via \eqref{optimal_activeIPC} in which $\mu$ is returned by BS is the global optimal solution of $P5$ given $\tilde{\bs}$.
\end{prop}
\begin{proof}
The proof is similar with the proof of Theorem 2 in \cite{Pan-20}. Denote the global optimal  solution for $P5$  as $\hat{\bs}$, then 
\bal
\label{1}
\tilde{h}_P(\hat{\bs}/\tilde{\bs})=\tilde{P}_I.
\eal
Then, denote the optimal $\mu$ returned by BS as $\hat{\mu}$, and the corresponding solution of $\bs$ given $\hat{\mu}$ as $\bs(\hat{\mu})$, one observes that $\hat{\mu}$ satisfies
\bal
\label{2}
\tilde{h}_P(\bs(\hat{\mu})/\tilde{\bs})=\tilde{P}_I.
\eal
Now  Proposition 1 can be proved via contradiction.  Assume that $\bs(\hat{\mu})$ is not the global optimal solution of $P5$, then we have $-\tilde{f}(\bs(\hat{\mu})/(u,\tilde{\bs}))\leq -\tilde{f}(\hat{\bs}/(u,\tilde{\bs}))$.
Note that $\bs(\hat{\mu})$ is the global optimal solution of $P5''$ under $\mu=\hat{\mu}$, then 
\bal
\notag
&-\tilde{f}(\bs(\hat{\mu})/(u,\tilde{\bs}))-\mu\tilde{h}_P(\bs(\hat{\mu})/\tilde{\bs})\\
\label{4}
\geq& -\tilde{f}(\hat{\bs}/(u,\tilde{\bs}))-\mu\tilde{h}_P(\hat{\bs}/\tilde{\bs}).
\eal
 By substituting \eqref{1} and \eqref{2} into \eqref{4} and after some manipulations, one obtains $-\tilde{f}(\bs(\hat{\mu})/(u,\tilde{\bs}))\geq -\tilde{f}(\hat{\bs}/(u,\tilde{\bs}))$, from which   $-\tilde{f}(\bs(\hat{\mu})/(u,\tilde{\bs}))= -\tilde{f}(\hat{\bs}/(u,\tilde{\bs}))$ follows
so that $\bs(\hat{\mu})=\hat{\bs}$, i.e., $\bs(\hat{\mu})$ in which $\hat{\mu}$ is returned by BS is the global optimal solution of $P5$. 
\end{proof}

Using this proposition, a global optimal solution of $\bs$ can also be obtained for $P5$ under active IPC case.  Therefore, during each iteration of SCA algorithm, $\bs$ is optimized either via \eqref{optimal_inactive} or \eqref{optimal_activeIPC}. As the convergence of SCA algorithm  is reached, a first-order optimal $\bs$ given fixed $u$ in $P4$ can be obtained.

Hence,  once the optimal $u$ returned by BS  is obtained, the corresponding $\bs$ is a sub-optimal solution for original $P1$ given fixed $\bw$. The Dinkelbach based method in combination with SCA and PB algorithm for optimizing $\bs$ given $\bw$ in $P3$ is summarized as Algorithm 1. In this algorithm,  there are three loops from the outer layer to inner layer. The  outer-most loop is the BS algorithm for optimizing $u$, the second middle layer loop is the SCA algorithm for solving $P5$, and the inner layer loop is the BS algorithm for optimizing $\mu$ in $P5''$ if IPC is tight and active.

\begin{algorithm}[h]
	\caption{(\it Algorithm for optimizing $\bs$ in $P3$)}
	\begin{algorithmic}
		\State \bRequire\  $u_l$, $u_u$, $\mu_l$, $\mu_u$.
		\Repeat\ \ 
     \State 1.  Initialize  $\tilde{\bs}$, set $u=(u_l+u_u)/2$, start BS algorithm.
           \Repeat\ \ 
             \State 2.  If IPC is inactive, optimize $\bs$ according to \eqref{optimal_inactive}, and go to step 4, else go to step 3.
              \Repeat\ \ 
              \State 3. Set  $\mu=(\mu_l+\mu_u)/2$, optimize $\bs(\mu)$ via \eqref{optimal_activeIPC}. If $\tilde{g}(\bs(\mu)/\tilde{\bs})\geq \tilde{P}_I$, set $\mu_l=\mu$, otherwise set $\mu_u=\mu$.
               \Until{$|\mu_u-\mu_l|$ converges}
          \State 4. Set $\bs$ as new starting point $\tilde{\bs}$.
		\Until{ $\tilde{f}(\bs/(u, \tilde{\bs}))$ converges} 
		\State 5. Output $\bs$ and compute $f(\bs/u)$.
		\State 6. If $f(\bs/u)\geq 0$, set $u_l=u$, otherwise set $u_u=u$.
		\Until{$|u_u-u_l|$ converges}
	\end{algorithmic}
\end{algorithm}

Finally, since $\bw$ and $\bs$ are optimized alternatively in AO algorithm, 
$C(\bw_1,\bs_1)\leq C(\bw_2,\bs_2)\leq...\leq C(\bw_k,\bs_k)$ 
where $C(\bw_k,\bs_k)$ is the objective value of $P1$ and $\bw_k$, $\bs_k$ are the  solutions during the iteration $k$. Furthermore, since $\bw$ is bounded by the inequality TPC, and $\bs$ is bounded by the equality IPC,    $C(\bw_{k},\bs_{k})$ is guaranteed to converge to a limit point $C(\bw_{opt},\bs_{opt})$. In the AO algorithm, the main computational complexity on optimizing $\bw$ given $\bs$ is about $O(m^{3.5})$ \cite{Boyd-04}. Given $\bw$, the main computational complexity of  computing
$\gl_1(\tilde{\bH}_{IE}-u\tilde{\bH}_{IB})$ 
and $\gl_1(\tilde{\bH}_{IP})$ when optimizing $\bs$ are about $O(n^3)$, and the main computational complexity in each iteration of SCA with PB approach is about $O(n^2)$.

\section{Algorithm for maximizing secrecy rate under imperfect Eve's CSI}
In this section, we consider a second assumption that the CSI of  Eve is imperfectly known at Alice due to estimation errors. Denote  $\tilde{\bh}_{AE}$ and $\tilde{\bH}_E$ as the estimated channels and $\tilde{\Delta}_E$, $\tilde{\Delta}_{AE}$ as the estimation errors, then the direct and cascaded channels between Alice and Eve are modeled as  $\bar{\bh}_{AE}=\tilde{\bh}_{AE}+\tilde{\Delta}_{AE}$, $\bar{\bH}_E=\tilde{\bH}_E+\tilde{\Delta}_E$ so that $\bh_{AE}=\hat{\bh}_{AE}+\Delta_{AE}$,  $\bH_E=\hat{\bH}_E+\Delta_E$ where $\hat{\bh}_{AE}=\tilde{\bh}_{AE}/\sigma_E$, $\hat{\bH}_E=\tilde{\bH}_E/\sigma_E$,  $\Delta_{AE}=\tilde{\Delta}_{AE}/\sigma_E$, $\Delta_E=\tilde{\Delta}_E/\sigma_E$. In particular, we assume that the errors are bounded as $\parallel \Delta_E \parallel_F=\parallel \tilde{\Delta}_E \parallel_F/\sigma_E\leq\tilde{\epsilon}_E/\sigma_E=\epsilon_E$, $\parallel \Delta_{AE} \parallel=\parallel \tilde{\Delta}_{AE} \parallel/\sigma_E\leq\tilde{\epsilon}_{AE}/\sigma_E=\epsilon_{AE}$, where $\tilde{\epsilon}_E$ and $\tilde{\epsilon}_{AE}$ are the uncertainty region of estimation errors known at Alice. 
 Based on these settings, the corresponding secrecy rate maximization problem  can be formulated as follows.
 \bal
\notag
(P6) \ &\underset{\bs, \bw}{\max}\underset{\Delta_{E}, \Delta_{AE}}{\min}\ C_B-C_E\\  
\notag
&s.t. \ \eqref{TPC+IPC}, \eqref{UMC}, \parallel \Delta_E \parallel_F\leq\epsilon_E, \parallel \Delta_{AE} \parallel\leq\epsilon_{AE}.
\eal
Note that this is  also a difficult non-convex problem with non-convex IPC, UMC and infinite non-convex bounded estimation error constraints. Before solving this problem, we firstly transform the object function and the constraints to a more tractable form. By dropping the log function in both $C_B$ and $C_E$ and adding an auxiliary variable $\tau$, $P6$ can be  equivalently transformed to
\bal
\notag
(P7) \ &\underset{\bs, \bw, \tau, \Delta_E, \Delta_{AE}}{\max}\ \frac{1+|(\bh_{AB}+\bs^{\rm H}\bH_B)\bw|^2}{1+\tau}\\
\notag
 s.t. &\ \eqref{TPC+IPC}, \eqref{UMC},  \tau\geq 0, |(\bh_{AE}+\bs^{\rm H}\bH_E)\bw|^2\leq \tau,\\
\label{infinite_constraints}
&\parallel \Delta_E \parallel_F\leq\epsilon_E, \parallel \Delta_{AE} \parallel\leq\epsilon_{AE}.
\eal
Then,  by  adopting Schur’s complement \cite{Boyd-04}, the constraint  $|(\bh_{AE}+\bs^H\bH_E)\bw|^2\leq \tau$  as well as the IPC in \eqref{TPC+IPC} can be equivalently transformed to 
\bal
\label{Schur}
\begin{bmatrix}
 \tau& (\bh_{AE}+\bs^{\rm H}\bH_E)\bw\\ 
  \bw^{\rm H}(\bh_{AE}^{\rm H}+\bH_E^{\rm H}\bs)& 1
\end{bmatrix}\geq\bo,\\
\label{SchurIPC}
\begin{bmatrix}
 P_I& (\bh_{AP}+\bs^{\rm H}\bH_P)\bw\\ 
  \bw^{\rm H}(\bh_{AP}^{\rm H}+\bH_P^{\rm H}\bs)& 1
\end{bmatrix}\geq\bo.
\eal
Then, substituting $\bH_E=\hat{\bH}_E+\Delta_E$, $\bh_{AE}=\hat{\bh}_{AE}+\Delta_{AE}$ into \eqref{Schur} and after some manipulations, one obtains that
\bal
\notag
&\begin{bmatrix}
 \tau& (\hat{\bh}_{AE}+\bs^{\rm H}\hat{\bH}_E)\bw\\ 
  \bw^{\rm H}(\hat{\bh}_{AE}^{\rm H}+\hat{\bH}_E^{\rm H}\bs)& 1
\end{bmatrix}\\
\notag
+&
\begin{bmatrix}
\bo_{1\times m}\\ \bw^{\rm H}
\end{bmatrix}
\begin{bmatrix}
\Delta_{AE}^{\rm H} & \bo_{m \times 1}
\end{bmatrix}
\bI+
\bI\begin{bmatrix}
\Delta_{AE}\\ \bo_{1 \times m}
\end{bmatrix}
\begin{bmatrix}
\bo_{m \times 1} & \bw
\end{bmatrix}
\\
\notag
+&
\begin{bmatrix}
\bo_{1\times m}\\ \bw^{\rm H}
\end{bmatrix}\Delta_E^{\rm H}
\begin{bmatrix}
\bs & \bo_{n \times 1}
\end{bmatrix}+
\begin{bmatrix}
\bs^{\rm H}\\ \bo_{1\times n}
\end{bmatrix}\Delta_E
\begin{bmatrix}
\bo_{m \times 1} & \bw
\end{bmatrix}\geq \bo.
\eal
To make further manipulations, we apply the following key lemma of general sign-definiteness principle, from which the proof can be seen in \cite{E-13}.
\begin{lemma}
Given matrices $\bA_i$, $\bB_i$, $i=1,2,...n$, and $\bY=\bY^{\rm H}$, the linear matrix inequality (LMI) 
$\bY\geq \sum_{i=1}^{n}(\bA_i^{\rm H}\bX_i\bB_i+\bB_i^{\rm H}\bX_i\bA_i), \forall i, \parallel\bX_i\parallel_F\leq\epsilon_i$  holds
only if there exists $u_i\geq0$, $\forall i$ such that
\bal
\notag
\begin{bmatrix}
 \bY-\sum_{i=1}^{n}u_i\bB_i^{\rm H}\bB_i&  -\epsilon_1\bA_1^{\rm H}& \cdots & -\epsilon_n\bA_n^{\rm H}\\ 
 -\epsilon_1\bA_1&  u_1\bI&  \cdots & \bo\\ 
 \vdots& \vdots &  \ddots& \vdots\\ 
 -\epsilon_n\bA_n& \bo & \cdots & u_n\bI
\end{bmatrix}\geq \bo.
\eal
\end{lemma}

Using this lemma, by introducing $u_1$ and $u_2$, and combing $\parallel \Delta_E \parallel_F\leq\epsilon_E, \parallel \Delta_{AE} \parallel\leq\epsilon_{AE}$, we obtain the LMI as
\bal
\notag
&\begin{bmatrix}
\tau-u_1n-u_2&  (\hat{\bh}_{AE}+\bs^{\rm H}\hat{\bH}_E)\bw& \bo_{1\times m} & \bo_{1\times m}\\ 
 \bw^{\rm H}(\hat{\bh}_{AE}^{\rm H}+\hat{\bH}_E^{\rm H}\bs)& 1-u_2&  \epsilon_E\bw^{\rm H} & \epsilon_{AE}\bw^{\rm H}\\ 
  \bo_{m\times 1}& \epsilon_E\bw &  u_1\bI&\bo_{m\times m}\\ 
 \bo_{m\times 1}& \epsilon_{AE}\bw & \bo_{m\times m} & u_2\bI
\end{bmatrix}\\
\label{LMIbound}
&\geq \bo.
\eal
Therefore, using these aforementioned manipulations, the infinite non-convex constraints in \eqref{infinite_constraints} and $|(\bh_{AE}+\bs^{\rm H}\bH_E)\bw|^2\leq \tau$ can be equivalently integrated to a single convex constraint \eqref{LMIbound} so that $P7$ can be equivalently transformed to 
\bal
\notag
(P8) \ &\underset{\bs, \bw, \tau, u_1, u_2}{\max}\ \frac{1+|(\bh_{AB}+\bs^{\rm H}\bH_B)\bw|^2}{1+\tau}\\
\notag
 &s.t. \ \parallel\bw\parallel^2\leq P_T , \eqref{UMC}, \eqref{SchurIPC}, \eqref{LMIbound}, \tau\geq 0, u_1\geq0, u_2\geq0.
\eal
We note that in $P8$, the secrecy rate is determined only by  $\bw, \bs$ and $\tau$, which are difficult  to optimize simultaneously. To optimize $P8$, in this paper, we firstly  obtain a proper region of $\tau$ numerically, and then apply AO algorithm to jointly optimize $\bw$ and $\bs$ by fixing $\tau$. Lastly, the optimal $\tau$ corresponding to the largest objective value in $P8$ can be found via searching algorithm in its region.

Specifically, the upper bound of $\tau$ can be formulated as 
\bal
\label{tau_area}
\tau\leq |(\bh_{AB}+\bs^{\rm H}\bH_B)\bw|^2\leq P_TJ(\bs)
\eal
where 
$J(\bs)=\parallel\bh_{AB}+\bs^{\rm H}\bH_B\parallel^2
=\bh_{AB}\bh_{AB}^{\rm H}+2Re\{\bs^{\rm H}\bH_{B}\bh_{AB}^{\rm H}\}+\bs^{\rm H}\bH_B\bH_B^{\rm H}\bs$.
In \eqref{tau_area}, the first inequality holds since the secrecy rate is non-negative,  and the second inequality follows from the lemma \cite{Li-13}: for any  $\bW\geq\bo$ and $tr(\bW)\leq P_T$, $\bh\bW\bh^{\rm H}\leq tr(\bW)\parallel\bh\parallel^2$ holds. Hence, the work reduces to maximize $J(\bs)$, which corresponds to the following problem
\bal
\notag
(P9) \ \underset{\bs}{\min}\ \bs^{\rm H}(-\bH_B\bH_B^{\rm H})\bs-2Re\{\bs^{\rm H}\bH_{B}\bh_{AB}^{\rm H}\}\ s.t.\ \eqref{UMC}.
\eal
  $P9$ can be maximized via SCA algorithm, in which the objective function can be approximated to a linear formula using Lemma 1 so that  a first-order optimal closed-form solution of $\bs$ can be obtained. Hence, assume that the optimized solution for $P9$ is $\bs_{opt}$,  the upper bound of $\tau$ can be finally obtained as $P_TJ(\bs_{opt})$. Therefore, $P8$ can be equivalently expressed as
\bal
\label{varphi}
\underset{\tau}{\max} \ \varphi(\tau) \ s.t.\  0\leq\tau\leq P_TJ(\bs_{opt})
\eal
where $\varphi(\tau)$ is defined as
\bal
\notag
(P10)&\ \varphi(\tau)\triangleq\underset{\bs, \bw, u_1, u_2}{\max}\ \frac{1+|(\bh_{AB}+\bs^{\rm H}\bH_B)\bw|^2}{1+\tau}\\ 
\notag
&s.t. \ \parallel\bw\parallel^2\leq P_T, \eqref{UMC}, \eqref{SchurIPC}, \eqref{LMIbound}, u_1\geq0, u_2\geq0.
\eal
To optimize $\bw$ and $\bs$ given $\tau$ in $P10$, we firstly fix $\bs$ and optimize $\bw$ in the following sub-problem
\bal
\notag
(P11)&\ \underset{\bw, u_1, u_2}{\max}\ |(\bh_{AB}+\bs^{\rm H}\bH_B)\bw|^2\\
\notag
 &s.t. \ \parallel\bw\parallel^2\leq P_T, \eqref{SchurIPC}, \eqref{LMIbound}, u_1\geq0, u_2\geq0.
\eal
Note that the objective function in $P11$ is not concave respect to $\bw$, however, it can be approximated via SCA. By applying the lemma that for any complex $x$ and $\tilde{x}$, 
\bal
\label{approximation}
|x|^2\geq 2Re\{\tilde{x}^*x\}-\tilde{x}^*\tilde{x}, 
\eal
$|(\bh_{AB}+\bs^{\rm H}\bH_B)\bw|^2$ can be lower bounded by
$2Re\{(\bh_{AB}+\bs^{\rm H}\bH_B)\tilde{\bw}\bw^{\rm H}(\bh_{AB}^{\rm H}+\bH_B^{\rm H}\bs)\}
 -(\bh_{AB}+\bs^{\rm H}\bH_B)\tilde{\bw}\tilde{\bw}^{\rm H}(\bh_{AB}^{\rm H}+\bH_B^{\rm H}\bs)$ 
where $\tilde{\bw}$ is a feasible point. Therefore, after dropping the constant term, the problem transforms to
\bal
\notag
(P11')&\ \underset{\bw, u_1, u_2}{\max}\ Re\{\bw^{\rm H}(\bh_{AB}^{\rm H}+\bH_B^{\rm H}\bs)(\bh_{AB}+\bs^{\rm H}\bH_B)\tilde{\bw}\}
\\
\notag
 &s.t.\  \parallel\bw\parallel^2\leq P_T, \eqref{SchurIPC}, \eqref{LMIbound}, u_1, u_2\geq0,
\eal
which can be directly solved via CVX solver since the objective function and all the constraints are convex. Hence, by solving $P10$ via SCA until convergence, a first-order optimal solution of $\bw$ given $\bs$ can be obtained.

The next step is to optimize $\bs$ given $\bw$ when fixing $\tau$ in the following problem
\bal
\notag
(P12) &\underset{\bs, u_1, u_2}{\min}\  r(\bs)=-|(\bh_{AB}+\bs^{\rm H}\bH_B)\bw|\\
\notag
 &s.t.\ \eqref{SchurIPC}, \eqref{LMIbound}, u_1\geq0, u_2\geq0, \eqref{UMC}.
\eal
Different from $P11$, $P12$  is difficult to solve, since although the objective function $r(\bs)$  can be approximated to linear expression, the non-convex  strict equality \eqref{UMC} can not be approximated.  Therefore, inspired by  \cite{Lipp-15}, we propose a P-CCP approach  to optimize $\bs$ in $P12$, in which the key idea is to  relax the problem by adding slack variables so as to make \eqref{UMC}  to be violated, and then penalizing the sum of the violations.  Note that a  significant difference between P-CCP and the aforementioned PB approach for solving $P5$ is that  in P-CCP, the non-convex \eqref{UMC} is relaxed to a convex one so that the optimized $|s_i|$ may not equal but infinitely close to 1, but in PB,   \eqref{UMC} is unchanged so that the solution $|s_i|=1$ always holds. 

Firstly, following \eqref{approximation}, 
$r(\bs)$ is upper bounded by 
\bal
\notag
r(\bs)\leq\tilde{r}(\bs,\tilde{\bs})= &-2Re\{(\bh_{AB}+\tilde{\bs}^{\rm H}\bH_B)\bw\bw^{\rm H}(\bh_{AB}^{\rm H}+\bH_B^{\rm H}\bs)\}\\
\notag
&+(\bh_{AB}+\tilde{\bs}^{\rm H}\bH_B)\bw\bw^{\rm H}(\bh_{AB}^{\rm H}+\bH_B^{\rm H}\tilde{\bs}).
\eal
Then,  by introducing auxiliary real vectors $\bb=[b_1, b_2,..., b_n]^T$, $\bc=[c_1, c_2,..., c_n]^T$ (where $b_i, c_i>0,  \forall i$), we formulate the following constraint
\bal
\label{violate}
 1-b_i\leq |s_i|^2\leq 1+c_i, \forall i=1,2,...,n.
\eal
so  that the feasible set of  $|s_i|$ changes to a continuous region. Although this violates the strict equality UMC, the sum of the violations $\sum_{i=1}^{n}b_i^k+ \sum_{i=1}^{n}c_i^k$ can be penalized so that the value of $b_i$ and $c_i$ ($i=1,2,...,n$) is guaranteed to converge to zero. Hence,  $|s_i|$  returned by P-CCP is infinitely close to 1. 

Note that \eqref{violate} can be further divided  into a convex constraint $ |s_i|^2\leq 1+c_i$ and non-convex constraint $|s_i|^2\geq 1-b_i$. The non-convex one can also be approximated (according to \eqref{approximation}) as $|s_i|^2\geq 2Re\{\tilde{s}_i^*s_i\}-|\tilde{s}_i|^2\geq 1-b_i$ where $\tilde{s}_i$ is a feasible point. 
Based on the key concept of P-CCP \cite{Lipp-15}, we introduce a penalty parameter $\gamma$ and construct the following optimization problem
\bal
\notag
(P12') &\underset{\bs, u_1, u_2, \bb, \bc}{\min}\  \tilde{r}(\bs,\tilde{\bs})+\gamma(\sum_{i=1}^{n}b_i+ \sum_{i=1}^{n}c_i)\\ 
\notag
s.t.&\ \eqref{SchurIPC}, \eqref{LMIbound}, u_1\geq0, u_2\geq0, b_i\geq 0, c_i\geq 0, \\
\notag
&2Re\{\tilde{s}_i^*s_i\}-|\tilde{s}_i|^2\geq 1-b_i,  |s_i|^2\leq 1+c_i, \forall i.
\eal
Therefore, given fixed $\gamma$, $P12'$ is a convex problem which can be directly solved via CVX. The algorithm of P-CCP for optimizing $\bs$ in $P12$ is summarized as Algorithm 2. In this algorithm, $\gamma$ is refreshed during each iteration, which is used to scale the impact of the penalty term $\sum_{i=1}^{n}b_i+ \sum_{i=1}^{n}c_i$ and control the feasibility of the constraints. And the upper limit $\gamma_{max}$  is used to avoid numerical problems if $\gamma$ grows too large and to provide convergence if a feasible region is not found. Following \cite{Lipp-15}, Algorithm 2 is not a strictly descent algorithm (mostly when $\gamma<\gamma_{max}$), but the value of objective function $r(\bs)+\gamma(\sum_{i=1}^{n}b_i+ \sum_{i=1}^{n}c_i)$ is always guaranteed to converge. The reason is that as $\gamma$ reaches to $\gamma_{max}$,  
\bal
\sum_{i=1}^{n}b_i^k+ \sum_{i=1}^{n}c_i^k \approx 0
\eal
where $b_i^k$ and $c_i^k$ are optimized at iteration $k$. Therefore, $b_i$ and $c_i$  both converge to zero so that the objective value of $P12'$ is only affected by  $\tilde{r}(\bs,\tilde{\bs})$. Furthermore, since $r(\bs)$ is also approximately upper bounded by $\tilde{r}(\bs,\tilde{\bs})$, one obtains that
\bal
\notag
&(r(\bs_k)+\gamma_{max}(\sum_{i=1}^{n}b_i^k+ \sum_{i=1}^{n}c_i^k))\\
\notag
-&(r(\bs_{k+1})+\gamma_{max}(\sum_{i=1}^{n}b_i^{k+1}+ \sum_{i=1}^{n}c_i^{k+1}))\leq \epsilon_{violation}
\eal
can be achieved where $\epsilon_{violation}$ is the target accuracy for convergence.
As the convergence is reached, the optimized solution satisfying \eqref{violate}  (where $b_i, c_i \approx 0$) returned by Algorithm 2 can be as an approximate first-order  optimal solution for the original problem $P12$.

\begin{algorithm}[h]
	\caption{(\it P-CCP algorithm for solving $P12$)}
	\begin{algorithmic}
		\State \bRequire\  $\gamma>0$, $\gamma_{max}>0$, $t>1$, $\tilde{\bs}$, set $k=0$.
		\Repeat\ \ 
     \State 1. Set $k=k+1$, optimize $\bs_{k}$ in the convex optimization problem $P12'$ given $\tilde{\bs}$.
         \State 2. Update the penalty parameter: $\gamma=\min(t\gamma, \gamma_{max})$.
          \State 3. Set $\tilde{\bs}=\bs_{k}$.
		\Until{The objective value $ r(\bs_k)+\gamma(\sum_{i=1}^{n}b_i^k+ \sum_{i=1}^{n}c_i^k)$ converges} 
	\end{algorithmic}
\end{algorithm}

Hence, using AO algorithm to jointly optimize $\bw$ and $\bs$ in the sub-problem $P11$ and $P12$ until converge, a limit point solution for the original  problem $P10$ can be obtained. Since $\tau$ lies in the interval $[0, P_TJ(\bs_{opt})]$, 
 the optimization problem in \eqref{varphi} can be solved by performing one-dimensional line search (such as uniform sampling or golden search \cite{Li-13}) over $\tau$, and choosing the optimal one $\tau_{opt}$ that achieves the maximum value of objective function.  In this paper, we apply uniform sampling method to search for  $\tau_{opt}$. Based on our extensive simulations, the searched point $\tau_{opt}$ is usually a very small value ranging from 0 to 1 in most channel realizations. Once  $\tau_{opt}$ is found, the corresponding $\bw$ and $\bs$ returned by AO algorithm is a final  limit point solution of original problem $P6$.  In this algorithm, given fixed $\tau$, the main computational complexity of optimizing $\bw$ given $\bs$ in each iteration of SCA and optimizing $\bs$ given $\bw$ in each iteration of P-CCP are $O((m+2)^{2})$ and $O((3n+2)^{2})$ respectively. And the main computational complexity of solving $P9$ via SCA in each iteration is about $O(n^2)$.

 In addition, we point out that when the CSI of $\bH_P$ and $\bh_{AP}$ between Alice and PR link are also imperfectly known due to bounded estimation errors as with the settings for $\bH_E$ and $\bh_{AE}$, our proposed AO algorithm shown above also can be extended to enhance the secrecy rate under this condition. This can be easily achieved by applying Schur's complement and Lemma 3 so that the infinite non-convex constraints can be directly transformed to single convex LMI constraint by adding auxiliary variables. Hence, AO algorithm can be directly applied to jointly optimize the beamformer and phase shift.

\section{AN aided Solution for enhancing secrecy rate under no Eve's CSI}

In this section, we consider a third assumption that the Eve's CSI is completely unknown at Alice so that both the direct channel $\bar{\bh}_{AE}$ and cascaded channel $\bar{\bH}_E$ cannot be obtained. Compared with the full and imperfect Eve's CSI cases in the previous sections, no Eve's CSI is a more practical assumption since Eve is usually a hidden passive malicious user, and it does not actively exchange its CSI with Alice. Therefore, it is unlikely to formulate an optimization problem as $P1$ or $P6$ to maximize the secrecy rate. Although the transmission rate at Bob  can be optimized given full CSI of Bob and PR, the information leakage at Eve may be large so that $C_B<C_E$, i.e.,  secure communication may not achievable according to the key concept of information-theoretic PLS.

Therefore, following our previous work for the non-CR setting in \cite{Dong-20d}\cite{Dong-20c}, in this paper, we propose an AN aided scheme to enhance the secrecy rate under no Eve's CSI. To achieve secure communication (i.e., $C_B>C_E$), the key idea of this scheme is summarized as two aspects. Firstly, we optimize a minimum power $P_S$ subject to the IPC at PR and a  QoS constraint at Bob in which the signal-to-noise ratio (SNR) meets the lowest pre-defined threshold. Secondly, we use the residual power $P_T-P_S$ available at Alice to transmit AN signals so as to decrease the SNR at Eve.

Based on this aforementioned setting, the signals received at Bob and Eve are represented as
\bal
\notag
&y_B= (\bar{\bh}_{AB}+\bh_{IB}diag(\bs^*)\bar{\bH}_{AI})(\bw x+\bz)+\xi_B, \\
\notag
&y_E= (\bar{\bh}_{AE}+\bh_{IE}diag(\bs^*)\bar{\bH}_{AI})(\bw x+\bz)+\xi_E
\eal
respectively where $\bz$ is the AN signals. In order to determine how much power of the total power $P_T$ can be dominated for transmitting the information signal $x$ satisfying the QoS constraint at Bob, we formulate a power minimization problem   as
\bal
\notag
(P13)\ \underset{\bw,\bs}{\min}\ \parallel\bw\parallel^2\ s.t.\ \eqref{SchurIPC}, |(\bh_{AB}+\bs^{\rm H}\bH_{B})\bw|^2\geq T, \eqref{UMC},
\eal
where $T$ is the pre-defined lowest SNR threshold at Bob. To solve this problem, we apply AO algorithm again to jointly optimize $\bw$ and $\bs$ in two sub-problems.  

Firstly, when $\bs$ is fixed, the sub-problem for optimizing $\bw$ is expressed as
\bal
\notag
&(P14)\ \underset{\bw}{\min}\ \parallel\bw\parallel^2\\
 \notag
&s.t.\ |(\bh_{AP}+\bs^{\rm H}\bH_{P})\bw|^2\leq P_I, |(\bh_{AB}+\bs^{\rm H}\bH_{B})\bw|^2\geq T.
\eal
Note that except for QoS constraint $|(\bh_{AB}+\bs^{\rm H}\bH_{B})\bw|^2\geq T$, the objective function and IPC are all convex in $\bw$. Using \eqref{approximation}, the QoS constraint  can be approximated to a linear formular   so that a first-order optimal $\bw$ can be obtained via SCA algorithm. Furthermore, it is straightforward to obtain that the  solution of $\bw$  makes the QoS constraint always hold with equality.

Secondly, when $\bw$ is fixed, the sub-problem of optimizing $\bs$ is expressed as
\bal
\notag
&\quad(P15)\  \rm{Find}\ \bs\quad s.t.\\
\notag
 &|(\bh_{AP}+\bs^{\rm H}\bH_{P})\bw|^2\leq P_I, |(\bh_{AB}+\bs^{\rm H}\bH_{B})\bw|^2\geq T, \eqref{UMC}.
\eal
Note that there is no objective function in this problem, and  any feasible $\bs$ satisfying the  QoS and UMC  can be as the optimal solution. In fact, if the feasible solution $\bs$ obtained for $P15$ achieves a strictly larger SNR than the target $T$, then the minimum transmit power in $P14$ returned by SCA can be properly reduced without violating the QoS constraint. Hence, the work can be reduced to maximize  the SNR at Bob  as large as possible in the  following problem $P15'$.
\bal
\notag
&(P15')\  \underset{\bs}{\max}\ |(\bh_{AB}+\bs^{\rm H}\bH_{B})\bw|^2\\
\notag
 &s.t.\  \eqref{UMC}, |(\bh_{AP}+\bs^{\rm H}\bH_{P})\bw|^2\leq P_I.
\eal 
Using the key idea of solving $P4$ given fixed $u$, $P15'$ also can be solved via SCA so as to obtain a first-order optimal solution, in which the objective function can be approximated via \eqref{approximation}  given a feasible point $\tilde{\bs}$ and the IPC $|(\bh_{AP}+\bs^{\rm H}\bH_{P})\bw|^2\leq P_I$ can be approximated via Lemma 1. And during each iteration of SCA algorithm, the global optimal solution of $\bs$ given $\tilde{\bs}$ can be obtained analytically under inactive IPC case and via the PB approach  under active  IPC case.  

Collectively, since $\bw$ and $\bs$ are optimized alternatively, the objective value $\parallel\bw\parallel^2$ is non-increasing with iteration $k$, i.e.,
$\parallel\bw_1\parallel^2\geq \parallel\bw_2\parallel^2\geq ... \geq \parallel\bw_k\parallel^2$.
Furthermore, since $\bw$ and $\bs$ are both bounded by the constraints, a limit point solution of $\bw$ and $\bs$ for $P13$ can be obtained as the AO algorithm converges. For optimizing $\bw$ given $\bs$ and optimizing $\bs$ given $\bw$ in each iteration of SCA, the main computational complexity are $O(m^2)$ and $O(n^2)$ respectively.

After obtaining the solution $\bw$, we obtain the  minimum power $P_S=\parallel\bw\parallel^2$ used for information signal transmission, then the residual  power $P_T-P_S$ is used to send AN signal $\bz$ to jam Eve. Since the Eve's channels are unknown at Alice, it is unlikely to optimize the transmit covariance of AN signal $\bR_{AN}=E\{\bz\bz^{\rm H}\}$. Therefore, we apply equal power allocation strategy to transmit AN signals to each dimension of $\mathcal{N}(\bH_{eff})$ where
\bal
\notag
\bH_{eff}\in\mathbb{C}^{m\times m}=\begin{bmatrix}
\bar{\bh}_{AB}+\bs^{\rm H}\bar{\bH}_B\\ \bh_{AP}+\bs^{\rm H}\bH_P
\end{bmatrix}^{\rm H}\begin{bmatrix}
\bar{\bh}_{AB}+\bs^{\rm H}\bar{\bH}_B\\ \bh_{AP}+\bs^{\rm H}\bH_P
\end{bmatrix} 
\eal
with $rank(\bH_{eff})=2$ so that the transmit covariance $\bR_{AN}$ is expressed as 
\bal
\notag
\bR_{AN}=\frac{P_T-P_S}{m-2}\bU_{AN}\bU_{AN}^{\rm H}
\eal
where the columns in the semi-unitary matrix $\bU_{AN}$ are all $m-2$ eigenvectors corresponding to zero eigenvalues of $\bH_{eff}$. Using this way, the AN generated by Alice could only interfere Eve, and does not affect the  quality of communication at both Bob and PR so that the QoS constraint at Bob and IPC at PR can be satisfied simultaneously.
Hence,  given target SNR threshold $T$, the  actual achievable secrecy rate $C_s$ by this scheme at Bob is finally expressed as 
\bal
\notag
&C_s=\log_2(1+T)\\
\notag
&-\log_2|1+\frac{|(\bar{\bh}_{AE}+\bs^{\rm H}\bar{\bH}_{E})\bw|^2}{\sigma_E^2+(\bar{\bh}_{AE}+\bs^{\rm H}\bar{\bH}_{E})\bR_{AN}(\bar{\bh}_{AE}+\bs^{\rm H}\bar{\bH}_{E})^{\rm H}}|.
\eal
Based on our extensive numerical simulations, a positive secrecy rate $C_s>0$ can be achieved via the proposed scheme in most channel realizations given finite $T$ and hence secure communication can be guaranteed.

\section{Simulation Results}


To validate the performance of our proposed numerical solutions,  extensive simulation results have been carried out in this section. Following \cite{Cui-19}, all the channels are assumed to be
independent Rayleigh fading, and each channel vector and matrix are  formulated as   the product of large scale fading and small scale fading.   For the location of each node, we consider a three dimensional coordinate space, and let Alice, Bob, IRS to be fixed in a coordinate, Eve and PR are randomly located in  certain areas (see Fig.1).  For all the AO algorithms illustrated in the paper, the starting point of $\bw$ and $\bs$ are randomly generated satisfying the corresponding constraints in each optimization problem. And the sampling interval for finding the optimal $\tau$ in the feasible interval $[0, P_TJ(\bs_{opt})]$  are set as $10^{-2}$.
Some main  parameter settings of channels, locations of each node as well as algorithms for the simulation are summarized in Table 1.  All the results plotted in Fig.2 to Fig.5 are averaged over 100   channel realizations, and all the results plotted in Fig.6 to Fig.9 are generated based on single   channel realization.

\begin{table*}[htbp]
\caption{Summary of Parameter Settings for Simulation}
\small 
\begin{tabular}{p{3 cm}|p{12.5cm}}
\hline
\hline
\rule{0pt}{10pt}\quad\quad\textbf{Symbol} & \quad\quad\quad\quad\quad\quad\quad\textbf{Definition and setting} \\
\hline
\rule{0pt}{10pt}\centering $\alpha_{AB}$, $\alpha_{AE}$, $\alpha_{AP}$ & The path loss exponents of  Alice-Bob, Alice-Eve and Alice-PR links respectively, all of them are set as 3. \\
\hline
\rule{0pt}{10pt}\centering $\alpha_{AI}$, $\alpha_{IB}$, $\alpha_{IE}$, $\alpha_{IP}$ &  The path loss exponents of  Alice-IRS, IRS-Bob, IRS-Eve and IRS-PR links respectively, all of them are set as 2.5.\\
\hline
\rule{0pt}{10pt}\centering $(0, 0, 0)$, $(100, 0, 0)$, $(50, 0, 50)$ & The  coordinates of Alice, Bob and IRS respectively, their locations are all fixed.\\
\hline
\rule{0pt}{10pt}\centering $(d_{x1}, d_{y1}, 0)$, $(d_{x2}, d_{y2}, 0)$ & The  coordinates of Eve and PR respectively.  PR is randomly  located in a circle with Alice as the center and a radius of 50 meters, i.e. $d_{x1}, d_{y1}\in[-50, 50]$. Eve is randomly located in a  circle with Bob as the center and a radius of 50 meters, i.e., $d_{x2}\in[50, 150]$,  $d_{y2}\in[-50, 50]$.\\
\hline
\rule{0pt}{10pt}\centering $\epsilon_{BS}$, $\epsilon_{AO}$, $\epsilon_{SCA}$& The target accuracy of all the bisection search algorithms,  AO algorithms and SCA algorithms, respectively. All these accuracy are set as $10^{-3}$. \\
\hline
\rule{0pt}{10pt}\centering  $t$ & The update of $\gamma$ in P-CCP algorithm, $t=5$.\\
\hline
\rule{0pt}{10pt}\centering $\gamma$, $\gamma_{max}$ & $\gamma=10$ and $\gamma_{max}=10^3$ are the initial and maximum value of penalty parameter respectively in P-CCP algorithm.\\
\hline
\rule{0pt}{10pt}\centering $\sigma_B^2$, $\sigma_E^2$ & The noise power at Bob and Eve respectively, they are both set as -100dBm.\\
\hline
\hline
\end{tabular}
\end{table*}

Fig.2 shows the  average secrecy rate performance returned by our proposed AO algorithm for the IRS-assisted design under full Eve's CSI. We also compare the results with those by other benchmark schemes: 1).
only optimizing $\bw$ given random $\bs$ at IRS; 2).optimal solutions in \cite{Pei-10} without IRS; 3).AN aided solutions in \cite{Fang-15} without IRS. Observe that the  performance via the proposed algorithm is significantly better than the those  benchmark schemes without IRS. The main reason is that IRS helps adjust the propagation channels, and via jointly optimizing $\bs$ and $\bw$, the signals transmitted via the direct link Alice-Bob (Eve) and reflecting link Alice-IRS-Bob (Eve) can be constructively (destructively) added at Bob (Eve), thereby boosting the secrecy rate. Although the random phase shift scheme has better performance than those without IRS, it still achieves significantly less performance than the proposed one. This indicates that only by the signal cooperation between Alice and IRS can the better performance gain be achieved.

\begin{figure}[t]
	\centerline{\includegraphics[width=3.0in]{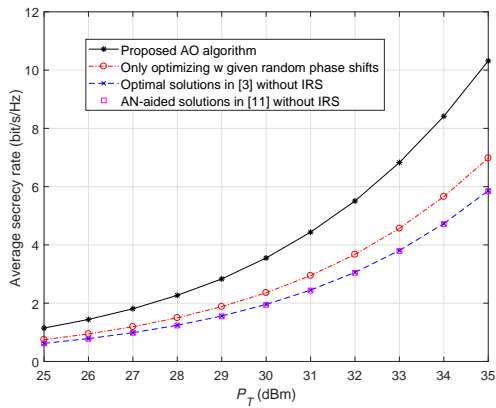}}
	\caption{The average secrecy rate performance versus $P_T$ returned by proposed algorithm with other benchmark schemes under full Eve's CSI.  $P_I$ is fixed at 30dBm, $m=4, n=8$.}
\end{figure}

To further exploit how the secrecy rate changes with $P_T$ as PR is located close to Alice under IRS-assisted and no IRS cases,   we set the location of PR  at $(20,0,0)$  between Alice and Bob so that the direct link of channels $\bh_{AB}$, $\bh_{AP}$ are strongly correlated, and the quality of $\bh_{AP}$ is better than $\bh_{AB}$ (since the signals transmitted via the channel $\bh_{AP}$ suffers from less attenuation), and we show the average secrecy rate as well as the average actual power consumption (i.e., the value of $||\bw||^2$) via the proposed algorithm for IRS-assisted case and existing solutions \cite{Pei-10} for no IRS case under this location setting.   Observe that in (a),  as  $P_T$ is increasing, the secrecy rate gradually saturate under no IRS case, since  the interference generated to PR gradually reach the maximum threshold $P_I$ so that Alice cannot allocate all the power for signaling. On the contrary, when IRS exists, the secrecy rate  increases with $P_T$ as if the IPC is relaxed. The main reason is that  IRS helps providing new communication link of Alice-IRS-Bob so that full power can be allocated  for signal transmission. This can be validated in (b), in which the red curve representing $||\bw||^2$ by the proposed algorithm coincides with the black curve representing total power $P_T$.  For no IRS case, $||\bw||^2<P_T$ returned by \cite{Pei-10} appears as $P_T>30$dBm, which indicates that full power allocation is not optimal due to tight IPC. Hence, we remark that IRS effectively eliminates the restriction brought by IPC on the unbounded growth of secrecy rate with $P_T$  even when the PR's channel is strongly corelated with Bob's channel.  which can never be achieved in the conventional CR system without IRS. 


\begin{figure}[t]
	\centerline{\includegraphics[width=3.0in]{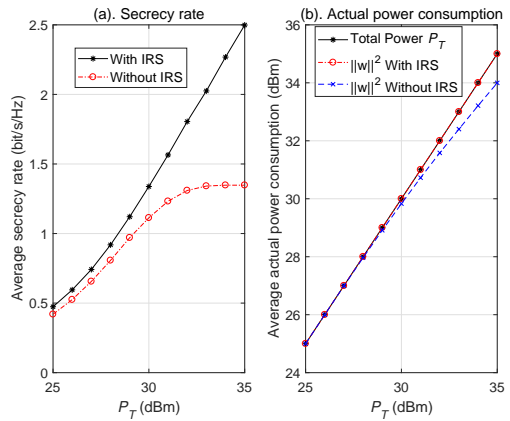}}
	\caption{The average secrecy rate and average actual power consumption returned by proposed algorithm with IRS and existing optimal solution in \cite{Pei-10} without IRS  when PR is located at (20, 0, 0).  $P_I$ is fixed at 25dBm, $m=n=3$.}
\end{figure}

Fig.4 shows the robust average secrecy rate returned by the proposed AO algorithm under imperfect Eve's CSI. Same with the results illustrated in Fig.2, our proposed solution also achieves significantly better performance than that for no IRS case as well as random phase shifts solution with IRS under different  channel estimation errors. And with larger settings of $m, n$, larger secrecy rate can be achieved due to the increased degree of freedom. Furthermore, as expected, the achievable secrecy rate is decreased when  the estimation error of Eve's CSI exists compared with perfect CSI case, since the signal power from the direct Alice-Bob link and reflected Alice-IRS-Bob link cannot accurately focus on Bob.  And as the region of estimation error is larger, more reduction in secrecy rate appears due to more  power loss when transmitting and reflecting the signals. Apart from this simulation, we also have made comparison about the secrecy rate between IRS-assisted and no IRS design by setting the location of PR at (20, 0, 0) given imperfect Eve's CSI.  Same with the results illustrate in Fig.3, full power allocation is always optimal under IRS-assisted design  so that secrecy rate keeps growing with $P_T$, from which the results are omitted here for brevity.

\begin{figure}[t]
	\centerline{\includegraphics[width=3.0in]{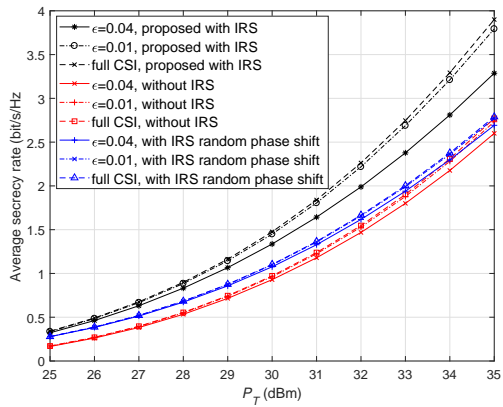}}
	\caption{The robust  average secrecy rate  returned by proposed algorithm with IRS, random phase shift with IRS and no IRS case under imperfect Eve's CSI condition. $\tilde{\epsilon}_E=\tilde{\epsilon}_{AE}=\epsilon$, $P_I$ is fixed at 30dBm, $m=3, n=4$.}
\end{figure}

Fig.5 shows the actual achievable average secrecy rate versus target SNR $T$ at Bob returned by the proposed AN aided scheme with IRS, random phase shift with IRS and  no IRS case  when  Eve's CSI is completely unknown. Based on the results, one observes that given a finite SNR ranging from 25dBm to 35dBm,  a positive secrecy rate $C_s>0$  still can be achieved by our scheme so that secure communication also can be guaranteed. Also note that  the  performance returned by our proposed solution with IRS is better than the other two solutions. The main reason is that more sufficient  power can be saved for Alice to achieve the same target SNR via jointly optimizing $\bw$ and $\bs$  than that by simply optimizing $\bw$ given random phase shifts at IRS or without IRS. Hence, for IRS-assisted design, Alice can use more residual power for AN signaling so as to decrease the SNR at Eve as much as possible, resulting higher secrecy rate. Here we remark that based on our extensive simulations, as long as  $T$ is set properly, positive actual secrecy rate can be achieved in most scenarios. But sometimes the actual secrecy rate becomes zero especially when Eve is located close to Alice. Therefore, to guarantee positive secrecy rate under this unfavorable environment, Bob needs to lower the value of $T$  so that Alice could save more power used to  decrease the SNR at Eve via AN signaling.

\begin{figure}[t]
	\centerline{\includegraphics[width=3.0in]{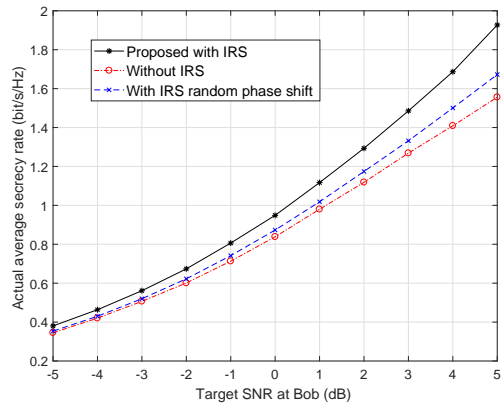}}
	\caption{The actual average secrecy rate  returned by the proposed AN aided scheme for IRS-assisted case, random phase shift for IRS-assisted case and no IRS case under no Eve's CSI. $P_T$ is fixed at 40dBm, $P_I$ is fixed at 30dBm. $m=n=4$.}
\end{figure}

To further exploit  the secrecy performance returned by our proposed scheme  when increasing the  target SNR $T$ to larger under no Eve's CSI case, Fig.6 shows the actual  secrecy rate $C_s$ versus $T$ under different settings of total power $P_T$. Note that for all settings of $P_T$, $C_s$ firstly increases with $T$  since the transmission rate $C_B=T$ at Bob dominates. However, as $T$ grows to higher, $C_s$ starts to decrease sharply, since the sufficient residual power $P_T-P_S$ for AN signaling is reduced significantly so that the information leakage to Eve $C_E$ dominates.  Finally, as $T$ grows to high enough,  the total
power $P_T$ can not support to meet the QoS constraint so that $P13$ becomes infeasible
and hence secure communication is not achievable (e.g., the secrecy rate stops at $T$=24 dB when $P_T\geq$55 dBm in the red curve). Therefore, we see that there is a trade off between improving QoS at Bob and enhancing the secrecy rate.   Hence, based on the results in Fig.5 and 6, one concludes that it is better to balance the setting between $T$ and the residual power $P_T-P_S$ for AN signaling so as to achieve a good secrecy performance.

\begin{figure}[t]
	\centerline{\includegraphics[width=3.0in]{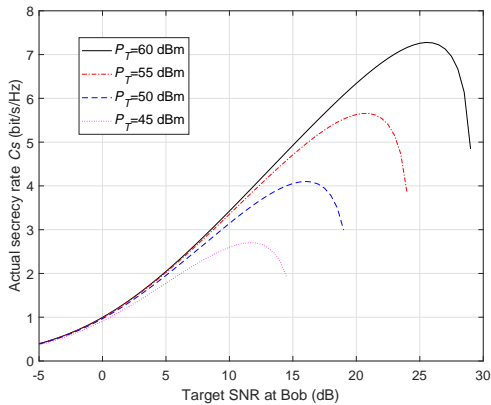}}
	\caption{The actual secrecy rate versus  $T$ returned by the proposed scheme under no Eve's CSI  given different total power $P_T$. $P_I$ is fixed at 30dBm. $m=4, n=6$.}
\end{figure}

 Fig.7 and Fig.8 illustrates the convergence of the objective function $C_s(\bw_k,\bs_k)$ in $P1$ under full CSI and imperfect Eve's CSI given fixed $\tau$ respectively. Observe that it requires about 4 to 16 iterations for $C_s(\bw_k,\bs_k)$ to converge to the target accuracy $10^{-3}$ under full CSI case, and 15 to 50 iterations are required to converge under imperfect Eve's CSI case. Also for both two CSI cases,  the convergence  is monotonically increasing. Based on our other extensive simulations, our proposed AO algorithms under both full and imperfect Eve's CSI  are guaranteed to monotonic convergence.

\begin{figure}[t]
	\centerline{\includegraphics[width=3.0in]{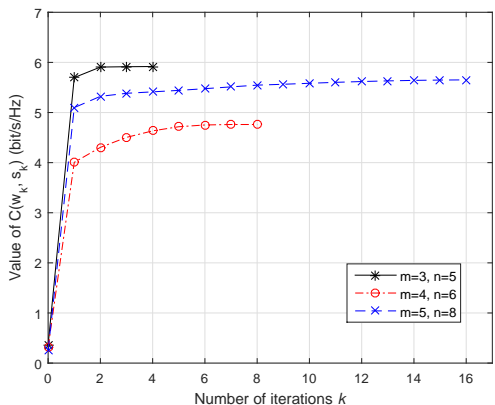}}
	\caption{Convergence of the objective value $C_s(\bw_k,\bs_k)$ in $P1$ returned by the proposed AO algorithm for solving $P1$ under full CSI. $P_T, P_I$ are fixed at 30dBm.}
\end{figure}

\begin{figure}[t]
	\centerline{\includegraphics[width=3.0in]{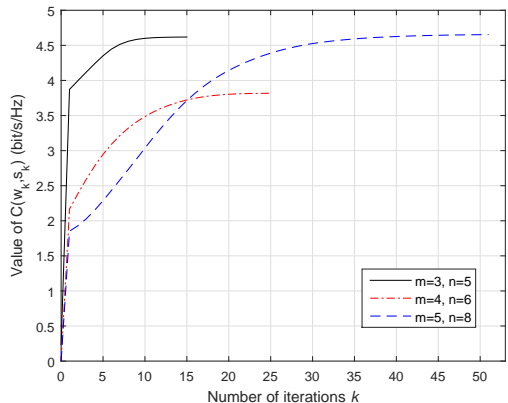}}
	\caption{Convergence of the objective value $C_s(\bw_k,\bs_k)$ in $P10$ returned by the proposed AO algorithm  given fixed $\tau$ under imperfect Eve's CSI. $P_T, P_I$ are fixed at 30dBm. $\tilde{\epsilon}_{E}=0.02, \tilde{\epsilon}_{AE}=0.005$.}
\end{figure}

Finally, Fig.9 gives the convergence of the proposed AO algorithm for solving power minimization problem $P13$ when no Eve's CSI is assumed. We plot the objective value $\parallel\bw_k\parallel^2$  in $P13$ versus the number of iterations $k$ under several randomly generated channels with different settings of $m, n$. Note that a monotonically non-increasing convergence is guaranteed for all the channel realizations. And same with the results in Fig.7 and Fig.8, larger settings of $m, n$ results in slower convergence since  more iterations is required to optimize  $\bw$ and $\bs$ with larger dimensions.

\begin{figure}[t]
	\centerline{\includegraphics[width=3.0in]{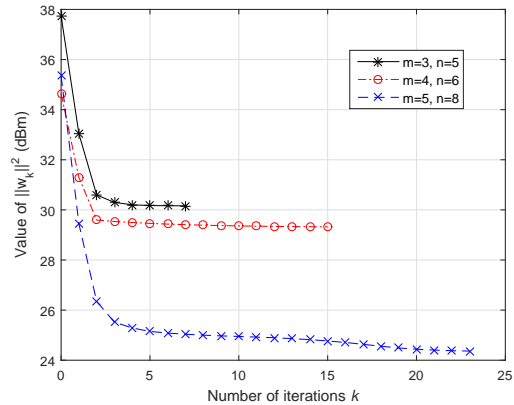}}
	\caption{Convergence of the objective value $\parallel\bw_k\parallel^2$ in $P13$ returned by the proposed AO algorithm given no Eve's CSI. $T$ is fixed at 0dB.}
\end{figure}
\section{Conclusion}

In this paper,  an IRS-assisted spectrum sharing underlay Gaussian CR WTC is studied, and we focus on enhancing the secrecy rate of this channel under  full CSI, imperfect Eve's CSI and completely no Eve's CSI cases. To solve the non-convex optimization problem, AO algorithm is proposed to jointly optimize the beamformer at secondary transmitter and phase shift vector at IRS. When no Eve's CSI is assumed, an AN aided approach is proposed to enhance the secrecy rate. Simulation results have shown that given full and imperfect CSI, our proposed solutions for the IRS-assisted design greatly enhance the secrecy rate compared with other benchmark solutions.  When no Eve's CSI is assumed, positive secrecy rate also can be achieved by our proposed AN aided scheme in most scenarios and hence secure communication can be guaranteed.


\begin{IEEEbiography}[{\includegraphics[width=1in,height=1.25in,		clip,keepaspectratio] {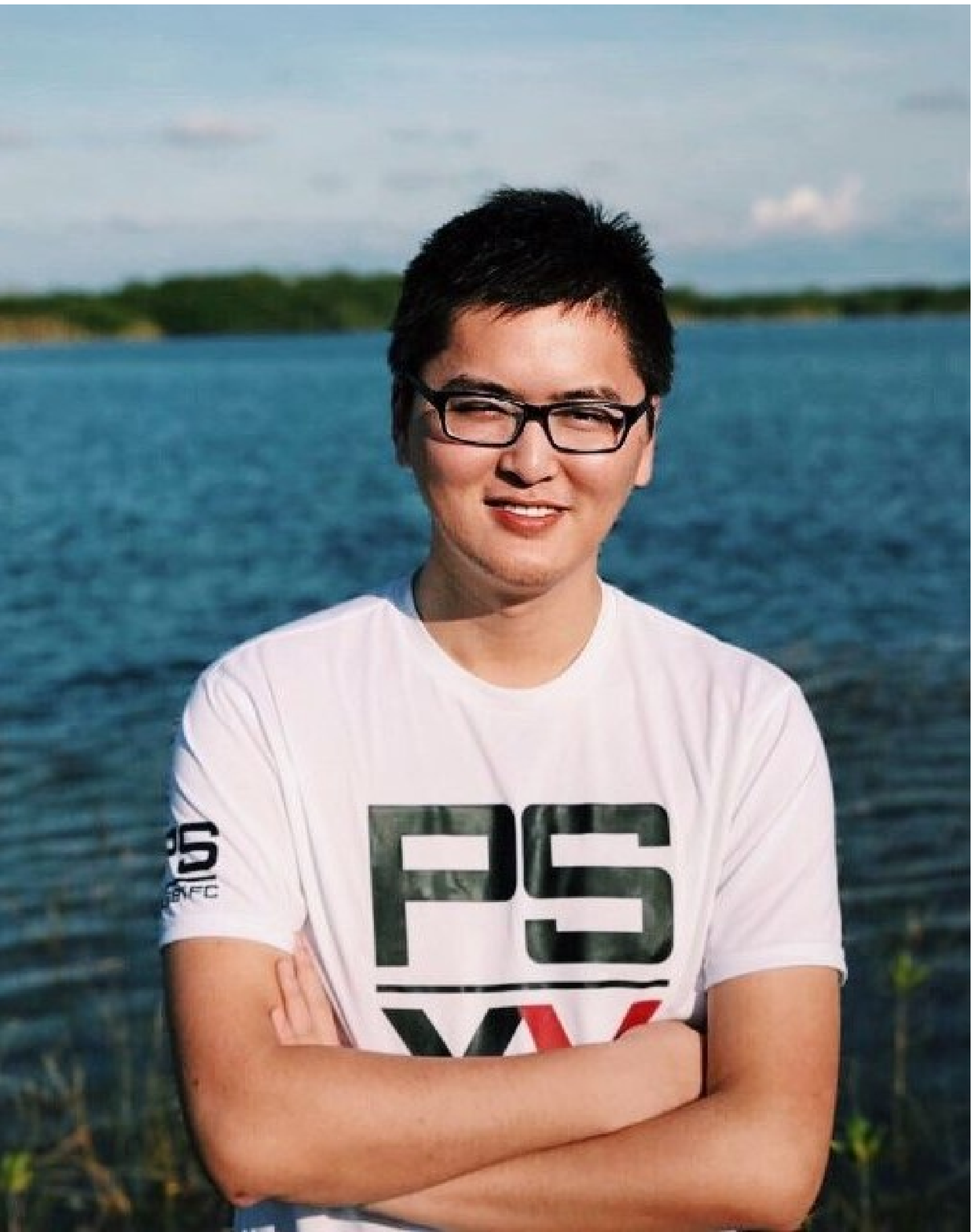}}]{Limeng Dong} was born in Xi’an, China. He received his bachelor, master as well as Ph.D degree from the School of Electronics and Information, Northwestern Polytechnical University, Xi’an, Shaanxi, 710072, China. He is now a postdoctoral researcher in the Ministry of Education Key Lab for Intelligent Networks and Network Security, School of Information and Communications Engineering, Xi’an Jiaotong University, Xi’an, Shaanxi, 710049, China. During 2015 to 2017, he was once a visiting Ph.D student at the School of Electrical Engineering and Computer Science, University of Ottawa, Canada. His research interests include multi-antenna communications, cognitive radio and physical layer security.
\end{IEEEbiography}

\begin{IEEEbiography}[{\includegraphics[width=1in,height=1.25in,clip,keepaspectratio]{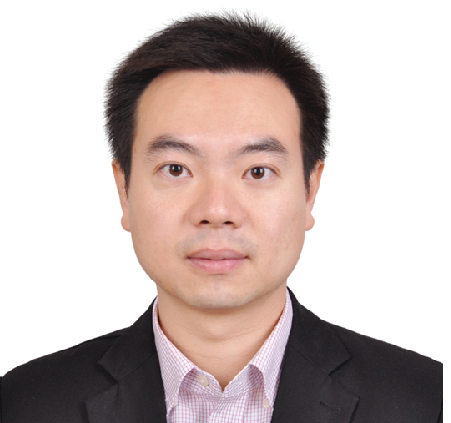}}]
{Hui-Ming Wang} (Senior Member, IEEE) received
the B.S. and Ph.D. degrees in electrical engineering from Xi’an Jiaotong University, Xi’an, China,
in 2004 and 2010, respectively.

From 2007 to 2008, and from 2009 to 2010, he was
a Visiting Scholar with the Department of Electrical
and Computer Engineering, University of Delaware,
Newark, DE, USA. He is currently a Full Professor
with Xi’an Jiaotong University. He has coauthored
the book \emph{Physical Layer Security in Random Cellular Networks} (Springer, 2016) and authored or
coauthored more than 150 IEEE journal articles and conference papers.
His research interests include 5G communications and networks, intelligent
communications, physical-layer security, and covert communications. He was
the Clarivate Highly Cited Researcher in 2019. He received the IEEE ComSoc
Asia–Pacific Best Young Researcher Award in 2018, the National Excellent
Doctoral Dissertation Award in China in 2012, and the Best Paper Award
from the IEEE/CIC International Conference on Communications in China
in 2014. He is also an Associate Editor of the
\textsc{IEEE Transactions on Communications}.

\end{IEEEbiography}

\begin{IEEEbiography}[{\includegraphics[width=1in,height=1.25in,clip,keepaspectratio]{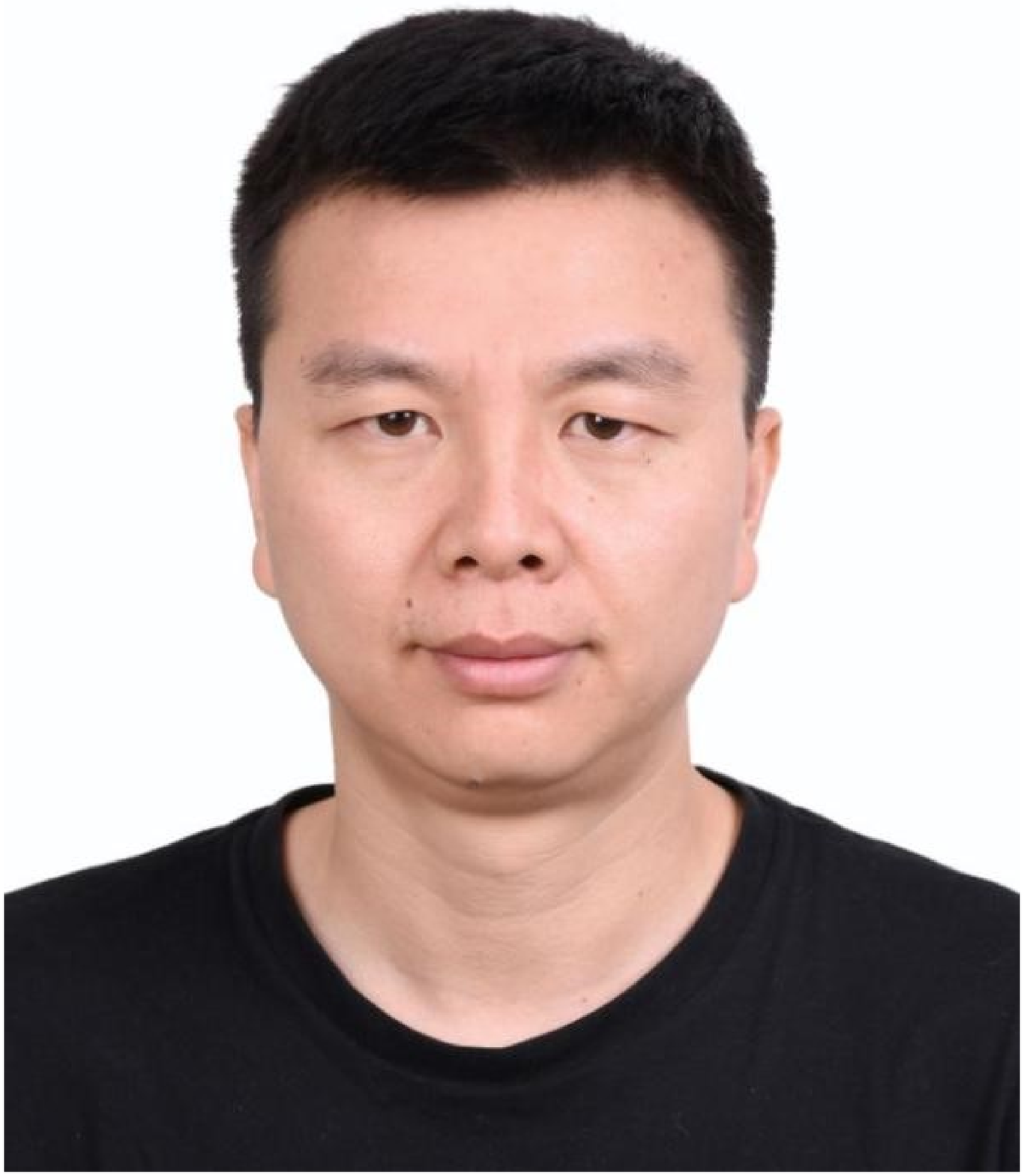}}]
{Haitao Xiao}  received his master and Ph.D degree from Waseda University, Tokyo, Japan. He is now an Assistant Professor in School of Information and Communication Engineering, Xi’an Jiaotong University, and Visiting Researcher of Graduate School of Information, Production and Systems, Waseda University.  
  From 2013 to 2015, he was a Assistant Researcher with the Research Center of of Information, Production and Systems, Waseda University, Japan. From 2016 to 2018, he was a Researcher with the Research Center of of Information, Production and Systems, Waseda University, Japan. His research is about Ad-Hoc, wireless Communication, Signal Process, Artificial intelligent, Covert Communication, Data Analysis and Bridge Diagnosis, etc. 
\end{IEEEbiography}



\begin{thebibliography}{99}



\bibitem {Fragkiadakis-13} A. G. Fragkiadakis \emph{et al.}, ``A survey on security threats and detection
techniques in cognitive radio networks,"  \emph{IEEE Commun. Surv. \& Tut.}, 
vol. 15, no. 1 , pp. 428-445, Feb. 2013.


\bibitem{Bloch-11} M. Bloch and J. Barros, ``Physical-layer security: from information theory to security engineering,"  \emph{Cambridge University Press}, 2011.


\bibitem{Pei-10} Y. Pei, Y.-C. Liang, L. Zhang, and K. C. Teh, ``Secure communication over MISO cognitive radio channels,"  \emph{IEEE Trans. Wireless Commun.}, vol. 9, no. 4, pp. 1494-1502, Apr. 2010.


\bibitem{Pei-11} Y. Pei, Y.-C. Liang, K. C. Teh, and K. H. Li, ``Secure communication in multiantenna cognitive radio networks with imperfect channel state information," \emph{IEEE Trans. Signal Process.}, vol. 59, no. 4, pp. 1683–1693, Apr. 2011.



\bibitem{Nahari-11} A. Al-Nahari \emph{et al.},  ``Beamforming with artificial noise for secure MISOME cognitive radio transmissions," \emph{IEEE Trans. Inf. Forens. Security}, vol. 13, no. 8, pp. 1875-1889, Aug. 2018. 




\bibitem{Nguyen-16} V-D. Nguyen \emph{et al.}, ``Joint information and jamming beamforming for secrecy rate maximization in cognitive radio networks,"  \emph{IEEE Trans. Inf. Forensics Security}, vol. 11, no. 11, pp. 2609-2623, Nov. 2016.



\bibitem{Wang-15} C. Wang, and H.-M. Wang, ``On the secrecy throughput maximization for MISO cognitive radio network in slow fading channels," \emph{IEEE Trans.  Info. Forensics Security}, vol. 9, no. 11, pp. 1814-1827, Nov. 2014.


\bibitem{Dong-18} L. Dong, S. Loyka and Y. Li, ``The secrecy capacity of gaussian MIMO wiretap channels under interference constraints,"  \emph{IEEE J. Sel. Areas Commun.}, vol. 36, no. 4, pp. 704-722, Apr. 2018.

\bibitem{Loyka-18} S. Loyka and L. Dong, ``Optimal full-rank signaling over MIMO wiretap channels under interference constraint," \emph{IEEE Wireless Commun. Letters}, vol. 7, no. 4, pp. 534-537, Aug. 2018.


\bibitem{Dong-20b} L. Dong, S. Loyka and Y. Li, ``Algorithms for globally-optimal secure signaling over Gaussian MIMO wiretap channels under interference constraints," \emph{IEEE Trans. Signal Process.}, vol. 68, pp. 4513-4528, 2020.


\bibitem {Fang-15}B. Fang, Z. Qian, W. Zhong, and W. Shao,  ``AN-aided secrecy precoding for SWIPT in cognitive MIMO broadcast channels," \emph{IEEE Commun. Letters}, vol. 19, no. 9, pp. 1632-1635, Sep. 2015.




\bibitem {Hu-18} S. Hu, F. Rusek, and O. Edfor, ``Beyond massive MIMO: the potential of data transmission with large intelligent surfaces," \emph{IEEE Trans.  Signal Process.}, vol. 66, no. 10, pp. 2746-2758, Mar. 2018.


\bibitem {Yuan-20} Y. Yuan, Y. Zhao, B. Zong, and S. Parolari ``Potential key technologies for 6G mobile communications," \emph{Science China-Information Sciences}, vol. 63, no. 8, pp. 217-235, Aug. 2020.


\bibitem {Renzo-20} M. D. Renzo  \emph{et al.},  ``Smart radio environments empowered by reconfigurable AI meta-surfaces: an idea whose time has come," \emph{EURASIP Journal on Wireless Communications and Networking}," vol. 1, pp. 1-20, 2019.


\bibitem {Renzo-20b} M. D. Renzo  \emph{et al.},  ``Smart radio environments empowered by reconfigurable intelligent surfaces: How it works, state of research, and road ahead," \emph{IEEE J. Sel. Areas Commun.}, vol. 38, no. 11, pp. 2450-2524, Nov. 2020.


\bibitem{Huang-20}  K.-W. Huang and H.-M. Wang, ``Passive beamforming for IRS aided wireless network,” \emph{IEEE Wireless Commun. Letters}, vol. 9, no. 12, pp. 2035-2039, Dec. 2020.




\bibitem {Wu-19b}  Q. Wu and R. Zhang, ``Intelligent reflecting surface enhanced wireless network via joint active and passive beamforming,"  \emph{IEEE Trans.  Wireless Commun.}, vol. 18, no. 11, pp. 5394-5409, Nov. 2019.


\bibitem{Pan-20} C. Pan \emph{et al.}, ``Intelligent reflecting surface aided MIMO bradcasting for simulataneous wireless information and power transfer," \emph{IEEE J. Sel. Areas Commun.}, vol. 38, no. 8, pp. 1719-1734, Aug. 2020.


\bibitem{Hu-20} X. Hu, J. Wang, and C. Zhong, ``Statistical CSI based design for intelligent reflecting surface assisted MISO systems," \emph{Science China-Information Sciences}, vol. 63, no. 12, pp. 211-220, Dec. 2020.


\bibitem {Zhou-19} G. Zhou \emph{et al.}, ``Robust beamforming design for intelligent
reflecting surface aided MISO communication systems,"  \emph{IEEE Wireless Commun. Letters}, vol. 9, no. 10, pp. 1658-1662, Oct. 2020.


\bibitem {Huang-19c} C. Huang, A. Zappone, G. C. Alexandropoulos, M. Debbah and C. Yuen,  ``Reconfigurable intelligent surfaces for energy efficiency in wireless communication," \emph{IEEE Trans. Wireless Commun.}, vol. 18, no. 8, pp. 4157-4170, Aug. 2019.


\bibitem {Yuan-19} J. Yuan \emph{et al.}, ``Intelligent reflecting surface-assisted cognitive
radio system," \emph{IEEE Trans. Commun.}, vol. 69, no. 1, pp. 675-687, Jan. 2021.


\bibitem {Guan-20} X. Guan, Q. Wu, and R. Zhang, ``Joint power control and passive beamforming in
IRS-assisted spectrum sharing," \emph{IEEE Commun. Letters}, vol. 24, no. 7, pp. 1553-1557, Jul. 2020.


\bibitem {Xu-20} D. Xu, X. Yu, and R. Schober, ``Resource allocation for intelligent reflecting
surface-assisted cognitive radio networks," \emph{2020 IEEE 21st International Workshop on Signal Processing Advances in Wireless Communications (SPAWC)}, Atlanta, GA, USA, May 2020.


\bibitem {Zhang-20} L. Zhang \emph{et al.},  ``Intelligent reflecting surface aided MIMO cognitive radio systems," \emph{IEEE Trans. Veh. Technol.}, vol. 69, no. 10, pp. 11445-11457, Oct. 2020.


\bibitem {He-20} J. He, K. Yu, Y. Zhou, and Y. Shi, ``Reconfigurable intelligent surface enhanced cognitive radio networks," \emph{arXiv:2005.10995}, [Online] https://arxiv.org/abs/2005.10995, 2020.


\bibitem {Zhang-20b} L. Zhang \emph{et al.},  ``Robust beamforming design for intelligent reflecting surface aided cognitive radio systems with imperfect cascaded CSI,"  \emph{arXiv:2004.04595}, [Online] https://arxiv.org/abs/2004.04595, 2020.


\bibitem {Xu-20b} D. Xu \emph{et al.}, ``Resource allocation for IRS-assisted full-duplex cognitive radio systems," \emph{IEEE Trans. Commun.}, vol. 68, no. 12, pp. 7376-7394,  Dec. 2020.



\bibitem {Cui-19} M. Cui, G. Zhang, and R. Zhang, ``Secure wireless communication via intelligent reflecting surface,"  \emph{IEEE Wireless Commun.n Letters}, vol. 8, no. 5, pp. 1410-1414, Oct. 2019.


\bibitem{Shen-19} H. Shen,  \emph{et al.}, ``Secrecy rate maximization for intelligent reflecting surface assisted multi-antenna communications,"  \emph{IEEE Commun. Letters}, vol. 23, no. 9, pp. 1488-1492, Jun. 2019.


\bibitem {Guan-19} X. Guan, Q. Wu, and R. Zhang, ``Intelligent reflecting surface assisted secrecy
communication: is artificial noise helpful or not?," \emph{IEEE Wireless Commun. Letters},  vol. 9, no. 6, pp. 778–782, Jun. 2020.



\bibitem {Yu-19} X. Yu and R. Schober, ``Enabling secure wireless communications via intelligent reflecting surfaces,"  \emph{2019 IEEE Global Communications Conference (GLOBECOM)}, Waikoloa, Hi, USA,  pp. 1-6, Dec. 2019.


\bibitem {Dong-21} L. Dong, H.-M. Wang, J. Bai, and H. Xiao, ``Double intelligent reflecting surface for secure transmission with inter-surface signal reflection," \emph{IEEE Trans. Veh. Technol.}, vol. 70, no. 3, pp. 2912-2916, Mar. 2021.


\bibitem{Dong-20d} H.-M. Wang, J. Bai, and L. Dong, ``Intelligent reflecting surface assisted secure transmission without eavesdropper's CSI," \emph{IEEE Signal Process. Letters}, vol. 27, pp. 1300-1304, 2020.


\bibitem {Xu-19} D. Xu, \emph{et al.}, ``Resource allocation for secure IRS-assisted multiuser MISO systems,"  \emph{2019 IEEE Globecom Workshops (GC WKshps)}, Waikoloa, HI, USA, pp. 1-6, Dec. 2019.


\bibitem{Yu-20} X. Yu \emph{et al.}, ``Robust and secure wireless communications via intelligent reflecting surfaces," \emph{IEEE J. Sel. Areas Commun.}, vol. 38, no. 11, pp. 2637-2652, Nov. 2020.


\bibitem {Dong-20} L. Dong and H.-M. Wang, ``Secure MIMO transmission via intelligent reflecting surface," \emph{IEEE Wireless Commun. Letters}, vol. 9, no. 6, pp.
787-790, Jun. 2020.


\bibitem {Dong-20c} L. Dong and H.-M. Wang, ``Enhancing Secure MIMO transmission via intelligent reflecting surface," \emph{IEEE Trans. Wireless Commun.}, vol. 19, no. 11, pp. 7543-7556, Nov. 2020.



\bibitem {Wang-19} Z. Wang, L. Liu, and S. Cui, ``Channel estimation for intelligent reflecting
surface assisted multiuser communications:
framework, algorithms, and analysis," \emph{IEEE Trans. Wireless Commun.},  vol. 19, no. 10, pp. 6607-6620, Oct. 2020.

 
\bibitem{Dinkelbach-67} W. Dinkelbach, ``On nonlinear fractional programming,"  \emph{Management
Science}, vol. 13, no. 7, pp. 492–498, Mar. 1967.


\bibitem{Sun-17} Y. Sun, P. Babu, and D. P. Palomar,  ``Majorization-minimization algorithms in signal processing, communications, and machine learning," \emph{IEEE Trans. Signal Process.}, vol. 65, no. 3, pp. 794-816, Feb. 2017.


\bibitem{Song-15} J. Song, P. Babu, and D. P. Palomar, ``Optimization methods for designing sequences with low autocorrelation sidelobes," \emph{IEEE Trans. Signal Process.}, vol. 63, no. 15, pp. 3998-4009,  Aug. 2015.


\bibitem{Boyd-04} S. Boyd and L. Vandenberghe, \emph{Convex Optimization}, Cambridge University Press,  2004.




\bibitem {E-13} E. A. Gharavol and E. G. Larsson, ``The sign-definiteness lemma and its applications to robust transceiver optimization
for multiuser MIMO systems,"  \emph{IEEE Trans. Signal Process.}, vol. 61, no. 2, pp. 238–252, Jan. 2013.


\bibitem {Li-13} Q. Li and W.-K. Ma, ``Spatially selective artificial-noise aided transmit optimization for MISO multi-eves secrecy rate maximization," \emph{IEEE Trans. Signal Process.}, vol. 61, no. 10, pp. 2704-2717, May 2013.



\bibitem {Lipp-15} T. Lipp and S. Boyd, ``Variations and extension of the convex–concave
procedure," \emph{Optimization \& Engineering}, vol. 17, no. 2, pp. 263-287, Nov. 2015.

\end{thebibliography}
\end{document}